\documentclass[11pt]{article}
\usepackage{fullpage}
\usepackage{fixltx2e}
\usepackage{amsmath,amsthm}
\usepackage{amssymb}
\usepackage{times}
\usepackage{mathrsfs}
\usepackage[latin1]{inputenc}
\usepackage{framed}
\usepackage[ruled,vlined]{algorithm2e}
\usepackage{multirow}
\usepackage{tikz}
\usetikzlibrary{arrows,automata,shapes,decorations,calc,matrix}
\newcommand{\val}{\mbox{\rm val}}
\newcommand{\cost}{\mbox{\rm r}}
\newcommand{\bb}{\mbox{\rm Num}}

\newcommand{\N}{\ensuremath{{\rm \mathbb N}}}

\newcommand{\Almost}{\mathsf{Almost}}

\newcommand{\Positive}{\mathsf{Positive}}

\newcommand{\Pred}{\mathsf{Pred}}
\newcommand{\Preprocessing}{\mbox{\rm Process}}

\newcommand{\Remove}{\mbox{\rm Remove}}

\newcommand{\act}{A}
\newcommand{\mov}{\Gamma}
\newcommand{\trans}{\delta}
\newcommand{\stra}{\sigma}
\newcommand{\bigstra}{\Sigma}
\newcommand\distr{{\mathcal D}}
\newcommand\pat{\omega}
\newcommand\pats{\Omega}
\newcommand{\game}{G}
\newcommand{\supp}{\mathrm{Supp}}
\newcommand{\dest}{\mathrm{Succ}}
\newcommand{\cala}{{\mathcal A}}
\newcommand{\LimInfAvg}{\mathsf{LimInfAvg}}
\newcommand{\LimSupAvg}{\mathsf{LimSupAvg}}
\def\set#1{\{ #1 \}}
\newcommand{\Safe}{\mathsf{Safe}}
\newcommand{\Reach}{\mathsf{Reach}}

\newcommand{\ov}{\overline}
\newcommand{\wh}{\widehat}
\newcommand{\Mem}{\mathsf{Mem}}
\newcommand{\mem}{\mathsf{m}}
\newcommand{\timedep}{\Theta}

\newcommand{\Allow}{\mbox{\rm Allow}_1}
\newcommand{\Bad}{\mbox{\rm Bad}_2}
\newcommand{\Good}{\mbox{\rm Good}_1}
\newcommand{\ASP}{\mbox{\rm ASP}}

\newcommand{\Aw}{\mathit{Aw}}
\newcommand{\Gd}{\mathit{Gd}}

\newcommand{\wt}{\widetilde}
\newcommand{\ial}{{\sc ImprovedAlgo}}

\makeatletter

\begingroup \catcode `|=0 \catcode `[= 1
\catcode`]=2 \catcode `\{=12 \catcode `\}=12
\catcode`\\=12 |gdef|@xcomment#1\end{comment}[|end[comment]]
|endgroup

\def\@comment{\let\do\@makeother \dospecials\catcode`\^^M=10\def\par{}}

\def\begincomment{\@comment\@xcomment}

\makeatother

 \newtheorem{theorem}{Theorem}

 \newtheorem{lemma}[theorem]{Lemma}
 \newtheorem{remark}[theorem]{Remark}

\begin{document}

\title{
Qualitative Analysis of Concurrent Mean-payoff Games\thanks{
The first author was supported by 
FWF Grant No P 23499-N23,  FWF NFN Grant No S11407-N23 (RiSE), ERC Start grant (279307: Graph Games), and Microsoft faculty fellows award. Work of the second author supported by the Sino-Danish Center for the Theory of Interactive Computation,
funded by the Danish National Research Foundation and the National
Science Foundation of China (under the grant 61061130540). The second author acknowledge support from the Center for research in
the Foundations of Electronic Markets (CFEM), supported by the Danish
Strategic Research Council.}}
\author{Krishnendu Chatterjee \thanks{IST Austria, \tt{krish.chat@ist.ac.at}}\and Rasmus Ibsen-Jensen \thanks{IST Austria, \tt{ribsen@ist.ac.at}}}
\date{}
\maketitle

\begin{abstract}
We consider concurrent games played by two-players on a finite-state 
graph, where in every round the players simultaneously choose a move,
and the current state along with the joint moves determine the successor state.
We study the most fundamental objective for concurrent games, namely, 
mean-payoff or limit-average objective, where a reward is associated to each transition,
and the goal of player~1 is to maximize the long-run average of the rewards,
and the objective of player~2 is strictly the opposite (i.e., the games are
zero-sum).
The path constraint for player~1 could be qualitative, i.e., the mean-payoff 
is the maximal reward, or arbitrarily close to it; or quantitative, i.e., 
a given threshold between the minimal and maximal reward.
We consider the computation of the almost-sure (resp. positive) winning sets,
where player~1 can ensure that the path constraint is satisfied with probability~1 
(resp. positive probability).
Almost-sure winning with qualitative constraint exactly corresponds
to the question of whether there exists a strategy to ensure that the payoff is
the maximal reward of the game.
Our main results for qualitative path constraints are as follows:
(1)~we establish qualitative determinacy results that show that for every state
either player~1 has a strategy to ensure almost-sure (resp. positive) winning 
against all player-2 strategies, or player~2 has a spoiling strategy to falsify 
almost-sure (resp. positive) winning against all player-1 strategies;
(2)~we present optimal strategy complexity results that precisely characterize
the classes of strategies required for almost-sure and positive winning for
both players; and
(3)~we present quadratic time algorithms to compute the almost-sure and the
positive winning sets, matching the best known bound of the algorithms for 
much simpler problems (such as reachability objectives).
For quantitative constraints we show that a polynomial time solution for the 
almost-sure or the positive winning set would imply a solution to a long-standing 
open problem (of solving the value problem of turn-based deterministic mean-payoff games) that is not known 
to be solvable in polynomial time.
\end{abstract}
\section{Introduction}

\noindent{\bf Concurrent games.}
Concurrent games are played by two players (player~1 and player~2) on 
finite-state graphs for an infinite number of rounds.
In every round, both players independently choose moves (or actions), and the 
current state along with the two chosen moves determine the successor state. 
In \emph{deterministic} concurrent games, the successor state is unique;
in \emph{stochastic} concurrent games, the successor state is given by 
a probability distribution.
The outcome of the game (or a \emph{play}) is an infinite sequence of states
and action pairs.
These games were introduced in a seminal work by Shapley~\cite{Sha53}, and 
have been one of the most fundamental and well-studied game models in 
stochastic graph games.
An important sub-class of concurrent games are \emph{turn-based} games,
where in each state at most one player can choose between multiple moves
(if the transition is stochastic we have turn-based stochastic games, 
and if the transition is deterministic we have turn-based deterministic games).

\smallskip\noindent{\bf Mean-payoff (limit-average) objectives.}
The most well-studied objective for concurrent games is the \emph{limit-average}
(or mean-payoff) objective, where a reward is associated to every transition 
and the payoff of a play is the limit-inferior (or limit-superior) average 
of the rewards of the play. 
The original work of Shapley~\cite{Sha53} considered \emph{discounted} sum 
objectives (or games that stop with probability~1); and concurrent stochastic 
games with limit-average objectives (or games that have zero stop 
probabilities) was introduced by Gillette in~\cite{Gil57}.
The player-1 \emph{value} $\val(s)$ of the game at a state $s$ is the 
supremum value of the expectation that player~1 can guarantee for the 
limit-average objective against all strategies of player~2.
The games are zero-sum where the objective of player~2 is the opposite.
Concurrent limit-average games and many important sub-classes have received 
huge attention over the last five decades. 
The prominent sub-classes are turn-based games as restrictions of the 
game graphs, and \emph{reachability} objectives as restrictions of the 
objectives.
A reachability objective consists of a set $U$ of \emph{terminal} states 
(absorbing or sink states that are states with only self-loops),
and the set $U$ is exactly the set of states where out-going transitions 
are assigned reward~1 and all other transitions are assigned reward~0.
Many celebrated results have been established for concurrent limit-average
games and its sub-classes:
(1)~the existence of values (or determinacy or equivalence of switching of 
strategy quantifiers for the players as in von-Neumann's min-max theorem) for 
concurrent discounted games was established in~\cite{Sha53};
(2)~the existence of values (or determinacy) for concurrent reachability games
was established in~\cite{Eve57};
(3)~the existence of values (or determinacy) for turn-based stochastic 
limit-average games was established in~\cite{LigLip69};
(4)~the result of Blackwell-Fergusson established existence of values for
the celebrated game of Big-Match~\cite{BF68}; and
(5)~developing on the results of~\cite{BF68} and Bewley-Kohlberg on Puisuex 
series~\cite{BK76} the existence of values for concurrent limit-average games was 
established in~\cite{MN81}.
The decision problem of whether the value $\val(s)$ is at least a rational 
constant $\lambda$ can be decided in PSPACE~\cite{CMH08,HKLMT11}; and is 
\emph{square-root sum} hard even for concurrent reachability 
games~\cite{EY06}.\footnote{The square-root sum problem is an important problem from computational
geometry, where given a set of natural numbers $n_1,n_2,\ldots,n_k$, 
the question is whether the sum of the square roots exceed an integer $b$. 
The square root sum problem is not known to be in NP.}
The algorithmic question of the value computation has also been studied in 
depth for special classes such as ergodic concurrent games~\cite{HK66} 
(where all states can be reached with probability~1 from all other states);
turn-based stochastic reachability games~\cite{Con92}; and turn-based
deterministic limit-average games~\cite{EM79,ZP96,GKK88,Brim11}.
The decision problem of whether the value $\val(s)$ is at least a rational 
constant $\lambda$ lie in NP $\cap$ coNP both for turn-based stochastic 
reachability games and turn-based deterministic limit-average games.
They are among the rare and intriguing combinatorial problems that lie 
in NP $\cap$ coNP, but are not known to be in PTIME.
The existence of polynomial time algorithms for the above decision questions 
are long-standing open problems.

\smallskip\noindent{\bf Qualitative winning modes.}
In another seminal work, the notion of \emph{qualitative winning} 
modes was introduced in~\cite{dAHK98} for concurrent reachability games.
In qualitative winning modes, instead of the exact value computation 
the question is whether the objective can be satisfied with 
probability~1 (\emph{almost-sure} winning) or with positive 
probability (\emph{positive} winning).
The qualitative analysis is of independent interest and importance in many 
applications (such as in system analysis) where we need to know whether the 
correct behaviour arises with probability~1.
For instance, when analysing a randomized embedded scheduler, we are
interested in whether every thread progresses with 
probability~1~\cite{EMSOFT05}.
Even in settings where it suffices to satisfy certain specifications with 
probability $p<1$, the correct choice of $p$ is a challenging problem, due 
to the simplifications introduced during modelling.
For example, in the analysis of randomized distributed algorithms it is 
quite common to require correctness with probability~1 
(see, e.g., \cite{PSL00,KNP_PRISM00,Sto02b}). 
More importantly it was shown in~\cite{dAHK98} that the qualitative analysis 
for concurrent reachability games can be solved in polynomial time 
(quadratic time for almost-sure winning, and linear time for positive
winning).
Moreover the algorithms were discrete graph theoretic algorithms, 
and the combinatorial algorithms were independent of the precise
transition probabilities. 
Since qualitative analysis is robust to numerical perturbations and modelling 
errors in the transition probabilities, and admits efficient combinatorial 
algorithms for the special case of concurrent reachability games, 
they have been studied in many different contexts such as Markov decision 
processes and turn-based stochastic games with $\omega$-regular objectives~\cite{CJH03,CH11,CH12}; 
pushdown stochastic games with reachability objectives~\cite{EY05,EY06,BBKO11}; 
and partial-observation games with $\omega$-regular objectives~\cite{CDHR07,BGG09,BGB12,CD12,CT12,CCT13,NV13,CDNV14}, 
to name a few. 
However, the qualitative analysis for the very important problem of 
concurrent limit-average games has not been studied before. 
In this work, we consider qualitative analysis of concurrent limit-average 
games. 
We show that the qualitative analysis of concurrent limit-average games 
is significantly different from and more involved than qualitative analysis of concurrent reachability 
games.

\smallskip\noindent{\bf Relevance of concurrent limit-average games.}
Besides the mathematical elegance of concurrent limit-average games, they also 
provide useful modeling framework for system analysis.
Concurrent games are relevant in modeling systems with synchronous interaction 
of components~\cite{AHM00a,AHM01a,AHK02}.
Mean-payoff objectives are widely used for performance measure of systems, 
such as in inventory control~\cite{FV97,Puterman}.
More recently, limit-average objectives have been used to ensure quality in 
synthesis of reactive systems~\cite{BCHJ09,CHJS10}, applied in synthesis of 
concurrent programs~\cite{CCHRS11}, and automata theoretic and 
temporal logic frameworks have been extended with such objectives to specify 
resource consumption requirements of systems~\cite{CDH10,BCHK11,DM12}.
Moreover, the LTL synthesis problem has also been extended with mean-payoff 
objectives~\cite{BBFR13}.
Thus the qualitative analysis problem for concurrent limit-average games is 
a relevant problem for formal analysis of systems.

\smallskip\noindent{\bf Classes of strategies.}
We first classify the various notion of strategies that are relevant for concurrent games.
In general a strategy in a concurrent game, considers the past 
history of the game (the finite sequence of states and actions played so 
far), and specifies a probability distribution over the next actions.
Thus a strategy requires memory to remember the past history of the 
game.
A strategy is \emph{stationary} if it is independent of the past history
and only depends on the current state; 
and a strategy is \emph{positional} if it is stationary and does not 
use randomization.
The complexity of a stationary strategy is described by its \emph{patience}
which is the inverse of the minimum non-zero probability assigned to a move.
The notion of patience was introduced in~\cite{Eve57} and also studied in the
context of concurrent reachability games~\cite{HKM09,HIM11}.
A strategy is \emph{Markov} if it only depends on the length of the play 
and the current state.
An infinite-memory strategy can be of different complexities, e.g., it could be 
implemented by a counter with increments (such as Markov strategies), or it could 
depend in a complicated way on the history such as strategies in Big-Match~\cite{BF68}.
To obtain a finer characterization of infinite-memory strategies we consider
the \emph{time-dependent} memory bound for them, which intuitively captures
the memory requirement as a function of the number of steps of the history.
For an infinite-memory strategy, the time-dependent memory needed is the amount 
of memory required for the first $T$ rounds of the game, for $T>0$.
For example, the time-dependent memory required by a Markov strategy is 
$T$, for all $T>0$.
We first show with an example the difference between concurrent reachability 
games and concurrent limit-average games for qualitative analysis.

\smallskip\noindent{\bf Example.}
Consider the classical game of \emph{matching penny} where in every 
round player~1 and player~2 choose independently between two moves, 
namely, heads and tails, and player~1 wins if the moves of both the
players match in any round. 
The game of matching penny is modelled as a concurrent reachability 
game with two states $s_0$ and $s_1$, where $s_1$ is the terminal 
state. 
In $s_0$, both players choose heads and tails, and if they match the 
successor state is $s_1$, otherwise $s_0$.
A stationary strategy for player~1 that chooses both moves
with equal probability is an almost-sure winning strategy.
Consider a variant of the matching penny game where player~1 wins 
immediately if the matching moves are tails, but if the matching moves
are heads, then player~1 gets a reward of~1 and the game continues.
The classical matching penny game and the variant matching penny game
are shown pictorially in Figure~\ref{fig:intro}.
For every $\epsilon>0$, the stationary strategy for player~1 that 
plays heads with probability $1-\epsilon$ and tails with 
probability~$\epsilon$ is an almost-sure winning strategy
for the objective to ensure that the limit-average payoff is
at least $1-\epsilon$.
For an almost-sure winning strategy for the objective to ensure
that the limit-average payoff is exactly~1, infinite-memory 
strategies are required, and a Markov strategy that in round~$j\geq 0$,
for $2^{j^2}$-steps plays tails with probability $\frac{1}{2^{j}}$ and 
heads with probability $1-\frac{1}{2^{j}}$, and then goes to round $j+1$, 
is an almost-sure winning strategy.
The variant matching penny game cannot be modeled as a 
concurrent reachability game.

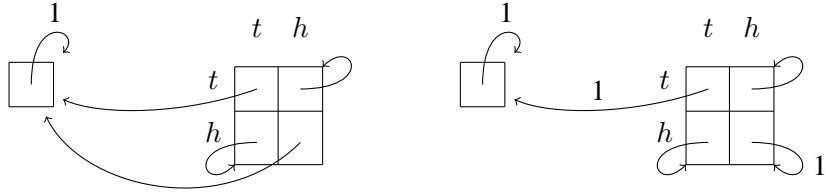
\begin{figure}
\centering
\begin{tikzpicture}[node distance=3cm]
\matrix (v2) [minimum height=1.5em,minimum width=1.5em,matrix of math nodes,nodes in empty cells, left delimiter={.},right delimiter={.}]
{
\\
};
\draw[black] (v2-1-1.north west) -- (v2-1-1.north east);
\draw[black] (v2-1-1.south west) -- (v2-1-1.south east);
\draw[black] (v2-1-1.north west) -- (v2-1-1.south west);
\draw[black] (v2-1-1.north east) -- (v2-1-1.south east);
\matrix (s2) [right of=v2,minimum height=1.5em,minimum width=1.5em,matrix of math nodes,nodes in empty cells, left delimiter={.},right delimiter={.}]
{
&t&h\\
t&&\\
h&&\\
};
\draw[black] (s2-2-2.north west) -- (s2-2-3.north east);
\draw[black] (s2-2-2.south west) -- (s2-2-3.south east);
\draw[black] (s2-3-2.south west) -- (s2-3-3.south east);
\draw[black] (s2-2-2.north west) -- (s2-3-2.south west);
\draw[black] (s2-2-3.north west) -- (s2-3-3.south west);
\draw[black] (s2-2-3.north east) -- (s2-3-3.south east);
\draw[->](v2-1-1.center) .. controls +(90:1) and +(45:1) ..  node[midway,above] (x) {1} (v2);
\draw[->](s2-2-2.center) .. controls +(200:1) and +(-25:1) .. (v2);
\draw[->](s2-2-3.center) .. controls +(0:1) and +(45:1) ..   (s2-2-3);
\draw[->](s2-3-2.center) .. controls +(180:1) and +(-135:1) ..   (s2-3-2);
\draw[->](s2-3-3.center) .. controls +(-135:1.5) and +(-65:1.5) ..   (v2);

\matrix (v1) [right of=s2,minimum height=1.5em,minimum width=1.5em,matrix of math nodes,nodes in empty cells, left delimiter={.},right delimiter={.}]
{
\\
};
\draw[black] (v1-1-1.north west) -- (v1-1-1.north east);
\draw[black] (v1-1-1.south west) -- (v1-1-1.south east);
\draw[black] (v1-1-1.north west) -- (v1-1-1.south west);
\draw[black] (v1-1-1.north east) -- (v1-1-1.south east);
\matrix (s) [right of=v1,minimum height=1.5em,minimum width=1.5em,matrix of math nodes,nodes in empty cells, left delimiter={.},right delimiter={.}]
{
&t&h\\
t&&\\
h&&\\
};
\draw[black] (s-2-2.north west) -- (s-2-3.north east);
\draw[black] (s-2-2.south west) -- (s-2-3.south east);
\draw[black] (s-3-2.south west) -- (s-3-3.south east);
\draw[black] (s-2-2.north west) -- (s-3-2.south west);
\draw[black] (s-2-3.north west) -- (s-3-3.south west);
\draw[black] (s-2-3.north east) -- (s-3-3.south east);
\draw[->](v1-1-1.center) .. controls +(90:1) and +(45:1) ..  node[midway,above] (x) {1} (v1);
\draw[->](s-2-2.center) .. controls +(200:1) and +(-25:1) .. node[midway,above,in=-45] (x) {1} (v1);
\draw[->](s-2-3.center) .. controls +(0:1) and +(45:1) ..   (s-2-3);
\draw[->](s-3-2.center) .. controls +(180:1) and +(-135:1) ..   (s-3-2);
\draw[->](s-3-3.center) .. controls +(0:1) and +(-45:1) ..  node[midway,below,right] (x) {1} (s-3-3);

\end{tikzpicture}
\caption{The classical matching pennies game (left) and our variant matching pennies game (right),
where rewards~$1$ are annotated with the transition and all other rewards are~0.}\label{fig:intro}

\end{figure}

\begin{table}[h]
\begin{center}
  \begin{tabular}{| c | c | c | c | c  | }
    \hline
     & & Exact Qual.& Limit Qual. & Reachability \\
      \hline
    \multirow{4}{40 pt}{\centering Ensuring \\strategy}  & \multirow{2}{40 pt}{\centering Sufficient} & {\bf Markov} & {\bf Stationary} & Stationary  \\ 
             &             & {\bf Time-dep: $T$}   & {\bf Patience $\left(\frac{n\cdot m}{\epsilon}\right)^{n^{O(n)}}$} & Patience $m$ \\ \cline{2-5}
 & \multirow{2}{40 pt}{\centering Necessary} & {\bf Infinite memory} & {\bf Stationary} & Stationary  \\ 
         &           & {\bf Time-dep: $T$} & {\bf Patience $\left(\frac{1}{\epsilon}\right)^{2^{\Omega(n)}}$} & Patience $m$ \\ \hline
    \multirow{4}{40 pt}{\centering Spoiling \\strategy}  &  \multirow{2}{40 pt}{\centering Sufficient} & {\bf Markov}  & {\bf Markov} & Markov  \\ 
              &            & {\bf Time-dep: $T$} & {\bf Time-dep: $T$} & Time-dep: $T$ \\ \cline{2-5}
 & \multirow{2}{40 pt}{\centering Necessary} & {\bf Infinite memory}  & {\bf Infinite memory}   & Infinite memory \\ 
              &            & {\bf Time-dep: $T$} & {\bf Time-dep: $T$} & Time-dep: $T$ \\ \hline
  \end{tabular}
\end{center}
\caption{Strategy complexity for almost-sure winning for exact qualitative, limit qualitative constraints, and reachability objectives in concurrent games,
where $m$ is the number of moves and $n$ is the number of states. The results in boldface are new results included in the present paper.}\label{tab:almost}
\end{table}

\begin{table}[h]
\begin{center}
  \begin{tabular}{| c | c | c | c | c  | }
    \hline
     & & Exact Qual.& Limit Qual. & Reachability \\
      \hline
    \multirow{4}{40 pt}{\centering Ensuring \\strategy}  & \multirow{2}{40 pt}{\centering Sufficient} & {\bf Markov} & {\bf Markov}  & Stationary  \\ 
             &             & {\bf Time-dep: $T$}   & {\bf Time-dep: $T$} & Patience $m$ \\ \cline{2-5}
 & \multirow{2}{40 pt}{\centering Necessary} & {\bf Infinite memory} & {\bf Infinite memory} & Stationary  \\ 
         &           & {\bf Time-dep: $T$} & {\bf Time-dep: $T$} & Patience $m$ \\ \hline
    \multirow{4}{40 pt}{\centering Spoiling \\strategy}  & \multirow{2}{40 pt}{\centering Sufficient} & {\bf Stationary} & {\bf Stationary} & Positional \\ 
              &            & {\bf Patience $m$} & {\bf Patience $m$} & \\ \cline{2-5}
 & \multirow{2}{40 pt}{\centering Necessary} & {\bf Stationary}  & {\bf Stationary}  & Positional \\ 
              &            & {\bf Patience $m$}  & {\bf Patience $m$} &  \\ \hline
  \end{tabular}
\end{center}
\caption{Strategy complexity for positive winning for exact qualitative, limit qualitative constraints, and reachability objectives in concurrent games, 
where $m$ is the number of moves. The results in boldface are new results included in the present paper.}\label{tab:positive}
\end{table}

\smallskip\noindent{\bf Our results.}
First,  note that for limit-average objectives the rewards can be scaled
and shifted, and hence without loss of generality we restrict ourselves to
the problem where the rewards are between~$0$ and~$1$.
We consider three kinds of path constraints (or objectives):
(i)~\emph{exact qualitative constraint} that consists of the set of paths where the 
limit-average payoff is~1;
(ii)~\emph{limit qualitative constraint} that consists of the set of paths with 
limit-average payoff at least $1-\epsilon$, for all $\epsilon>0$; and
(iii)~\emph{quantitative constraint} that consists of the set of paths where the 
limit-average payoff is at least~$\lambda$, for $\lambda \in (0,1)$.
The significance of qualitative constraint are as follows: 
first, they present the most strict guarantee as path constraint;
and second, almost-sure winning with qualitative constraint exactly corresponds
to the question whether there exists a strategy to ensure that the payoff is
the maximal reward of the transitions of the game. 
Our results are as follows:

\begin{enumerate}

\item \emph{Almost-sure winning.}
Our results for almost-sure winning are as follows:
\begin{enumerate} 
\item \emph{(Qualitative determinacy).} 
First we establish (in Section~\ref{subsec1}) 
\emph{qualitative} determinacy for concurrent limit-average games 
for almost-sure winning where we show that
for every state either player~1 has a strategy to ensure almost-sure
winning for both exact qualitative constraint and limit
qualitative constraint against all player-2 strategies; 
or player~2 has a spoiling strategy to ensure that both exact 
qualitative constraint and limit qualitative constraint are violated
with positive probability against all player-1 strategies.
The qualitative determinacy result is achieved by characterizing the 
almost-sure winning set with a discrete combinatorial cubic-time algorithm. 

\item \emph{(Strategy complexity).}
In case of concurrent reachability games, stationary almost-sure winning 
strategies with patience $m$ (where $m$ is the number of moves) 
exist for player~1; and spoiling strategies for player~2 require infinite memory
and Markov strategies are sufficient~\cite{dAHK98}.
In contrast, we show that for exact qualitative path constraint, 
almost-sure winning strategies require infinite memory for player~1 
and Markov strategies are sufficient; whereas the spoiling strategies
require infinite memory for player~2 and Markov strategies are sufficient.
For limit qualitative constraint, we show that for all $\epsilon>0$, 
stationary almost-sure winning strategies exist for player~1,
whereas spoiling strategies for player~2 require infinite memory and 
Markov strategies are sufficient.
We establish asymptotically matching double exponential upper and lower 
bound for the patience required by almost-sure winning strategies 
for limit qualitative constraints.
In all cases where infinite-memory strategies are required we establish
that the optimal (matching upper and lower bound) time-dependent memory bound is $T$, for all $T>0$.
Our results are summarized in Table~\ref{tab:almost} (and the results
are in Section~\ref{subsec2}).

\item \emph{(Improved algorithm).}
Finally we present an improved algorithm for the computation of
the almost-sure winning set of exact and limit qualitative constraint
that uses quadratic time (in Section~\ref{subsec3}). 
Our algorithm matches the bound of the currently best known algorithm 
for the computation of the almost-sure winning set of the special case of  
concurrent reachability games.
\end{enumerate}

\item \emph{Positive winning.} Our results for positive winning are as follows:
\begin{enumerate}
\item \emph{(Qualitative determinacy and algorithm).}
We establish the qualitative determinacy for positive winning;
and our qualitative determinacy characterization already presents a 
quadratic time algorithm to compute the positive winning sets for exact
and limit qualitative path constraints.
Moreover, also for positive winning the exact and limit qualitative path 
constraints winning sets coincide.
The results are presented in Section~\ref{subsec4}.

\item \emph{(Strategy complexity).}
In case of concurrent reachability games, stationary positive winning 
strategies with patience $m$ (where $m$ is the number of moves) 
exist for player~1; and positional (stationary and deterministic) spoiling strategies
exist for player~2~\cite{dAHK98}.
In contrast, we show that positive winning strategies for player~1 both for exact and limit 
qualitative path constraints require infinite memory and Markov strategies are
sufficient, and the optimal time-dependent memory bound is $T$, for all $T>0$.
We also show that (a)~stationary spoiling strategies exist for player~2, 
(b)~they require randomization, and (c)~the optimal bound for patience is $m$.
Our results are summarized in Table~\ref{tab:positive} (and the results
are in Section~\ref{subsec5}).

\end{enumerate}

\item \emph{(Hardness of polynomial computability for quantitative constraints).}
Finally we show (in Section~\ref{Sec:quan}) that for quantitative path 
constraints, both the almost-sure and the 
positive winning problems even for turn-based stochastic mean-payoff games with rewards 
only $\set{0,1}$ are at least as hard as value computation of turn-based 
deterministic mean-payoff games with arbitrary integer rewards.
Thus solving the almost-sure or the positive winning problem with quantitative 
path constraint with boolean rewards in polynomial time would imply the 
solution of a long-standing open problem.
Observe that we show hardness for the almost-sure and positive winning in 
turn-based stochastic boolean reward games with quantitative constraints. 
Note that
(i)~turn-based deterministic boolean reward games with quantitative constraints
can be solved in polynomial time (the pseudo-polynomial time algorithm of~\cite{ZP96}
is polynomial for boolean rewards);
(ii)~almost-sure and positive winning for both turn-based stochastic and concurrent 
reachability games can be solved in polynomial time~\cite{CJH03,dAHK98};
and
(iii)~almost-sure and positive winning with qualitative constraints can be solved
in polynomial time as shown by our results (even for concurrent games). 
Thus our hardness result is tight in the sense that the natural restrictions 
in terms of game graphs, objectives, or qualitative constraints yield polynomial 
time algorithms. 

\end{enumerate}

\smallskip\noindent{\bf Important remarks.} 
Observe that for positive winning our algorithm uses quadratic time, as compared to the
linear time algorithm for positive reachability in concurrent games.
However, for the special case of turn-based deterministic games,
the positive winning set for exact qualitative path constraints 
coincide with the winning set for coB\"uchi games (where the goal is to ensure
that eventually always a set $T$ of states are visited).
The long-standing best known algorithms for turn-based deterministic coB\"uchi games 
uses quadratic time. Turn-based deterministic coB\"uchi games is a special case of positive winning for concurrent 
limit-average games with qualitative path constraints. 
Therefore our algorithm matches the current best known quadratic bound of the 
simpler case.
Finally, our results that for qualitative analysis Markov strategies
are sufficient are in sharp contrast to general concurrent limit-average 
games where Markov strategies are not sufficient (for example in the celebrated
Big-Match game~\cite{BF68}).

\section{Definitions}

In this section we present the definitions of game structures, strategies, 
objectives, winning modes and other basic notions.

\smallskip\noindent{\bf Probability distributions.}
For a finite set~$A$, a {\em probability distribution\/} on $A$ is a
function $\trans\!:A\to[0,1]$ such that $\sum_{a \in A} \trans(a) = 1$.
We denote the set of probability distributions on $A$ by $\distr(A)$. 
Given a distribution $\trans \in \distr(A)$, we denote by $\supp(\trans) = 
\{x\in A \mid \trans(x) > 0\}$ the {\em support\/} of the distribution 
$\trans$.

\smallskip\noindent{\bf Concurrent game structures.} 
A (two-player) {\em concurrent stochastic game structure\/} 
$\game =  (S, \act,\mov_1, \mov_2, \trans)$ consists of the 
following components.

\begin{itemize}

\item A finite state space $S$ and a finite set $\act$ of actions (or moves).

\item Two move assignments $\mov_1, \mov_2 \!: S\to 2^{\act}
	\setminus \emptyset$.  For $i \in \{1,2\}$, assignment
	$\mov_i$ associates with each state $s \in S$ the non-empty
	set $\mov_i(s) \subseteq \act$ of moves available to player $i$
	at state $s$.  
	For technical convenience, we assume that $\mov_i(s) \cap \mov_j(t) = \emptyset$ 
	unless $i=j$ and $s=t$, for all $i,j \in \{1,2\}$ and $s,t \in S$. 
	If this assumption is not met, then the moves can be trivially renamed to satisfy the assumption.

\item A probabilistic transition function
	$\trans\!:S\times\act\times\act \to \distr(S)$, which
	associates with every state $s \in S$ and moves $a_1 \in
	\mov_1(s)$ and $a_2 \in \mov_2(s)$ a probability
	distribution $\trans(s,a_1,a_2) \in \distr(S)$ for the
	successor state.
\end{itemize}
We will denote by $\trans_{\min}$ the minimum non-zero transition 
probability, i.e., $\trans_{\min}=\min_{s,t \in S} \min_{a_1 \in \mov_1(s),a_2\in \mov_2(s)}
\set{\trans(s,a_1,a_2)(t) \mid \trans(s,a_1,a_2)(t)>0}$.
We will denote by $n$ the number of states (i.e., $n=|S|$), and by 
$m$ the maximal number of actions available for a player at a state 
(i.e., $m=\max_{s\in S} \max\set{|\mov_1(s)|,|\mov_2(s)|}$).
For all states $s \in S$, moves $a_1 \in
\mov_1(s)$ and $a_2 \in \mov_2(s)$, we indicate by 
$\dest(s,a_1,a_2) = \supp(\trans(s,a_1,a_2))$ 
the set of possible successors of $s$ when moves $a_1$ and $a_2$
are selected. 
The size of the transition relation of a game structure is defined as
$|\delta|=\sum_{s\in S}\sum_{a_1 \in \mov_1(s)} \sum_{a_2 \in \mov_2(s)} |\dest(s,a_1,a_2)|$.

\smallskip\noindent{\bf Turn-based stochastic games, turn-based deterministic games 
and MDPs.} 
A game structure $\game$ is {\em turn-based stochastic\/} if at every
state at most one player can choose among multiple moves; that is, for
every state $s \in S$ there exists at most one $i \in \{1,2\}$ with
$|\mov_i(s)| > 1$. 
A turn-based stochastic game with deterministic transition function is 
a turn-based deterministic game.
A game structure is a player-2 \emph{Markov decision process (MDP)} if for all 
$s \in S$ we have $|\mov_1(s)|=1$, i.e., only player~2 has choice of 
actions in the game, and player-1 MDPs are defined analogously.

\medskip\noindent{\bf Plays.}
At every state $s\in S$, player~1 chooses a move $a_1\in\mov_1(s)$,
and simultaneously and independently
player~2 chooses a move $a_2\in\mov_2(s)$.  
The game then proceeds to the successor state $t$ with probability
$\trans(s,a_1,a_2)(t)$, for all $t \in S$. 
A {\em path\/} or a {\em play\/} of $\game$ is an infinite sequence
$\pat =\big( (s_0,a^0_1, a^0_2), (s_1, a^1_1, a^1_2), (s_2,a_1^2,a_2^2)\ldots\big)$ of states and action pairs such that for all 
$k\ge 0$ we have (i)~$a^k_1 \in \mov_1(s_k)$ and $a^k_2 \in \mov_2(s_k)$; and (ii)~$s_{k+1} \in \supp(\trans(s_k,a^k_1,a^k_2))$.
We denote by $\pats$ the set of all paths.


\smallskip\noindent{\bf Strategies.}
A {\em strategy\/} for a player is a recipe that describes how to 
extend prefixes of a play.
Formally, a strategy for player $i\in\{1,2\}$ is a mapping 
$\stra_i\!:(S\times \act \times \act)^* \times S \to\distr(\act)$ that associates with every 
finite sequence $x \in (S\times \act \times \act)^*$  of state and action pairs, and the 
current state $s$ in $S$, representing the past history of the game, 
a probability distribution $\stra_i(x \cdot s)$ used to select
the next move. 
The strategy $\stra_i$ can prescribe only moves that are available to player~$i$;
that is, for all sequences $x\in (S \times \act \times \act)^*$ and states $s\in S$, we require that
$\supp(\stra_i(x\cdot s)) \subseteq \mov_i(s)$.  
We denote by $\bigstra_i$ the set of all strategies for player $i\in\{1,2\}$.
Once the starting state $s$ and the strategies $\stra_1$ and $\stra_2$
for the two players have been chosen, 
the probabilities of events are uniquely defined~\cite{VardiP85}, where an {\em
event\/} $\cala\subseteq\pats$ is a measurable set of
paths.
For an event $\cala\subseteq\pats$, we denote by $\Pr_s^{\stra_1,\stra_2}(\cala)$ 
the probability that a path belongs to $\cala$ when the game starts from $s$ and 
the players use the strategies $\stra_1$ and~$\stra_2$.
We will consider the following special classes of strategies:

\begin{enumerate}

\item \emph{Stationary (memoryless) and positional strategies.} 
A strategy $\sigma_i$ is \emph{stationary} (or memoryless) if it is independent of the history 
but only depends on the current state, i.e., for all $x,x'\in(S\times A\times A)^*$ and all $s\in S$, 
we have $\sigma_i(x\cdot s)=\sigma_i(x'\cdot s)$, and thus can be expressed as a function 
$\stra_i: S \to \distr(\act)$. 
For stationary strategies, the complexity of the strategy is described by the 
\emph{patience} of the strategy, which is the inverse of the minimum non-zero 
probability assigned to an action~\cite{Eve57}. 
Formally, for a stationary strategy  $\stra_i:S \to \distr(\act)$ for player~$i$, 
the patience is  
$\max_{s \in S} \max_{a \in \mov_i(s)} \set{ \frac{1}{\stra_i(s)(a)} \mid \stra_i(s)(a)>0}$.
A strategy is \emph{pure (deterministic)} if it does not use randomization, i.e., for any history there is always
some unique action $a$ that is played with probability~1.
A pure stationary strategy $\stra_i$ is also called a {\em positional} strategy, and 
represented as a function $\stra_i: S \to \act$.

\item \emph{Strategies with memory and finite-memory strategies.} 
A strategy $\stra_i$ can be equivalently defined as a pair of functions $(\stra_i^u,\stra_i^n)$,
along with a set $\Mem$ of memory states, such that 
(i)~the next move function $\stra_i^n: S \times \Mem \to \distr(\act)$ given the current state of 
the game and the current memory state specifies the probability distribution over the actions; and
(ii)~the memory update function $\stra_i^u: S \times \act \times \act \times \Mem \to \Mem$ given the current state of the game, the action pairs, 
and the current memory state updates the memory state. 
Any strategy can be expressed with an infinite set $\Mem$ of memory states, and a strategy is 
a {\em finite-memory} strategy if the set $\Mem$ of memory states is finite, otherwise it is an 
{\em infinite-memory} strategy.

\item \emph{Markov strategies.}
A strategy $\stra_i$ is a \emph{Markov strategy} if it only depends on the length of
the play and current state. 
Formally, for all finite prefixes $x, x' \in (S\times \act \times \act)^*$ such that $|x|=|x'|$ (i.e., the 
length of $x$ and $x'$ are the same, where the length of $x$ and $x'$ are the number of 
states that appear in $x$ and $x'$, respectively) and all $s \in S$ 
we have $\stra_i(x\cdot s)=\stra_i(x' \cdot s)$.

\item \emph{Time-dependent memory.} 
Consider a strategy $\stra_i$ with memory $\Mem$. 
For every finite sequence $x \in (S \times \act \times \act)^*$ 
there is a unique memory element $t(x)=\mem \in \Mem$ 
such that after the finite sequence $x$ the current memory state is $\mem$
(note that the memory update function is a deterministic function).
For a time bound $T$, the time-dependent memory of the strategy $\stra_i$, 
is the size of the set of memory elements used for histories upto length
$T$, i.e., 
$|\set{\mem \in \Mem \mid \exists x \in (S\times \act \times \act)^*, |x|\leq T, 
t(x) =\mem }|$.
Formally, the time-dependent memory for an infinite-memory strategy $\stra_i$ is a 
function $\timedep_{\stra_i}: \N \to \N$ such that 
$\timedep_{\stra_i}(T)= |\set{\mem \in \Mem \mid \exists x 
\in (S\times \act \times \act)^*, |x|\leq T, t(x) =\mem }|$.
Note that a Markov strategy can be played with time-dependent memory of 
size $T$, for all $T>0$, i.e., for a Markov strategy $\stra$ we have 
$\timedep_{\stra}(T)=T$ for all $T>0$.
A trivial upper bound on the time-dependent memory of a strategy is 
$(|S|\cdot |A|\cdot |A|)^T$, for all $T \geq 0$, 
i.e., for all strategies $\stra$ we have 
$\timedep_{\stra}(T) \leq (|S|\cdot |A|\cdot |A|)^T$, for all $T \geq 0$.

\end{enumerate}

\smallskip\noindent{\bf Repeated games with absorbing states.}
A state $s$ is {\em absorbing} if for all actions $a_1\in \Gamma_1(s)$ and all actions $a_2\in \Gamma_2(s)$ we have $\dest(s,a_1,a_2)=\{s\}$. A game is a {\em repeated game with absorbing states}, as defined by Kohlberg~\cite{kohlberg74}, if all states, other than one special state $s^*$, are absorbing. In the present paper all absorbing states will only have a single action for each player. Once an absorbing state is reached, no strategy will need memory. In a repeated game with absorbing states, updates of memory will therefore only happen in state $s^*$ and implies that the play has only been in state $s^*$ since the start of the play. We will therefore write $\stra_i^u(s^*,a_1,a_2,\mem)$ as $\stra_i^u(a_1,a_2,\mem)$.

\smallskip\noindent{\bf Objectives.} 
An objective $\Phi \subseteq \pats$ is a measurable subset of paths. In this 
work we will consider \emph{limit-average} (or mean-payoff) objectives.
We will consider concurrent games with a \emph{boolean} reward function $\cost: S \times \act \times 
\act \to \set{0,1}$ that assigns a reward value $\cost(s,a_1,a_2)$ for all 
$s\in S$, $a_1 \in \mov_1(s)$ and $a_2 \in \mov_2(s)$ (see Remark~\ref{rem:rew} for 
general real-valued reward functions\footnote{We consider boolean rewards for simplicity in presentation,
and in Remark~\ref{rem:rew} we argue how the results extend to general rewards.}).
For a path $\pat= \big((s_0, a^0_1, a^0_2), (s_1, a^1_1,a^1_2), \ldots\big)$, 
the limit-inferior average (resp. limit-superior average) is defined as 
follows:
\[
\LimInfAvg(\pat)= \lim\inf_{n \to \infty} \frac{1}{n} \cdot \sum_{i=0}^{n-1} \cost(s_i,a^i_1,a^i_2); 
\quad
\LimSupAvg(\pat)= \lim\sup_{n \to \infty} \frac{1}{n} \cdot \sum_{i=0}^{n-1} \cost(s_i,a^i_1,a^i_2).
\]
For a threshold $\lambda \in [0,1]$ we consider the following objectives:

\begin{align*}
&\LimInfAvg(\lambda) = \set{\pat \mid \LimInfAvg(\pat) \geq \lambda}; 
&\quad \LimSupAvg(\lambda) = \set{\pat \mid \LimSupAvg(\pat) \geq \lambda};\\
&\overline{\LimInfAvg}(\lambda) = \set{\pat \mid \LimInfAvg(\pat) < \lambda}; 
&\quad \overline{\LimSupAvg}(\lambda) = \set{\pat \mid \LimSupAvg(\pat) < \lambda};\\
&\overline{\LimInfAvg}_{\leq}(\lambda) = \set{\pat \mid \LimInfAvg(\pat) \leq \lambda}; 
&\quad \overline{\LimSupAvg}_{\leq}(\lambda) = \set{\pat \mid \LimSupAvg(\pat) \leq \lambda}.\\
\end{align*}

For the analysis of concurrent games with boolean limit-average objectives 
we will also need \emph{reachability} and \emph{safety} objectives.
Given a target set $U\subseteq S$, the reachability objective $\Reach(U)$ requires 
some state in $U$ be visited at least once, i.e., defines the set
$\Reach(U)= \set{\pat=\big((s_0, a^0_1, a^0_2), (s_1, a^1_1,a^1_2), \ldots\big) 
\mid \exists i \geq 0. s_i \in U}$ of paths.
The dual safety objective for a set $F\subseteq S$ of safe states requires that the 
set $F$ is never left, i.e., 
$\Safe(F)= \set{\pat=\big((s_0, a^0_1, a^0_2), (s_1, a^1_1,a^1_2), \ldots\big)
\mid \forall i \geq 0. s_i \in F}$.
Observe that reachability objectives are a very special case of 
boolean reward limit-average objectives where states in $U$ 
are absorbing and are exactly the states with reward~1, and similarly
for safety objectives.


\smallskip\noindent{\bf $\mu$-calculus, complementation, and levels.} 
Consider a $\mu$-calculus expression $\Psi = \mu X. \psi(X)$ over a
finite set $S$, where $\psi: 2^S \mapsto 2^S$ is monotonic.
The least fixpoint $\Psi = \mu X. \psi(X)$ is equal
to the limit $\lim_{k \to \infty} X_k$, where $X_0 = \emptyset$,
and $X_{k+1} = \psi(X_k)$.  
For every state $s \in \Psi$, we define the {\em level\/} $k \geq 0$
of $s$ to be the integer such that $s \not\in X_k$ and $s \in X_{k+1}$.  
The greatest fixpoint $\Psi = \nu X. \psi(X)$ is
equal to the limit $\lim_{k \to \infty} X_k$, where $X_0 = S$, and
$X_{k+1} = \psi(X_k)$.   
For every state $s \not\in \Psi$, we define the {\em level\/} $k \geq 0$ of
$s$ to be the integer such that $s \in X_k$ and $s \not\in X_{k+1}$.  
The {\em height\/} of a $\mu$-calculus expression 
$\gamma X. \psi(X)$, where $\gamma \in \set{\mu, \nu}$, 
is the least integer $h$ such that $X_h = \lim_{k \to \infty} X_k$.
An expression of height $h$ can be computed in $h+1$ iterations. 
Given a $\mu$-calculus expression $\Psi=\gamma X. \psi(X)$, where 
$\gamma \in \set{\mu, \nu}$, the complement $\neg \Psi = (S \setminus \Psi)$ of
$\gamma$ is given by 
$\overline{\gamma} X. \neg \psi (\neg X)$, 
where $\overline{\gamma} = \mu$ if $\gamma = \nu$, and 
$\overline{\gamma} = \nu$ if $\gamma= \mu$.

\smallskip\noindent{\bf Almost-sure and positive winning sets.} 
Given an objective $\Phi$, the {\em almost-sure  winning set} for player~1 for
the objective $\Phi$, denoted as $\Almost_1(\Phi)$, is the set of 
states such that there exists a strategy (referred to as almost-sure winning strategy) 
for player~1 to ensure that the objective is satisfied with probability~1 (almost-surely) 
against all strategies of the opponent.
The {\em positive winning set}, denoted $\Positive_1(\Phi)$, requires that
player~1 can ensure that the probability to satisfy $\Phi$ is positive.
Formally we have 
\begin{itemize}
\item $\Almost_1(\Phi)=\set{s \mid \exists \stra_1. \forall \stra_2. \Pr_s^{\stra_1,\stra_2}(\Phi)=1}$; and
\item $\Positive_1(\Phi)=\set{s \mid \exists \stra_1. \forall \stra_2. \Pr_s^{\stra_1,\stra_2}(\Phi)>0}$.
\end{itemize}
The almost-sure and positive winning sets $\Almost_2$ and $\Positive_2$ for player~2 are 
obtained analogously, as above, by switching the roles of player~1 and player~2, respectively.

\section{Almost-sure Winning}\label{sec:almost}
In this section we will present three results: (1)~establish 
\emph{qualitative} determinacy for almost-sure winning; 
(2)~establish the strategy complexity for almost-sure winning;
and (3)~finally present an improved algorithm to compute the almost-sure
winning set;
for exact and limit qualitative constraints in concurrent games.

\subsection{Qualitative determinacy}\label{subsec1}
We will establish the qualitative determinacy results through a polynomial time 
algorithm to compute the set 
$\Almost_1(\LimInfAvg(\lambda))$ and $\Almost_1(\LimSupAvg(\lambda))$ for $\lambda=1$
in concurrent games with boolean rewards. 
To present our algorithm we first define a three-argument predecessor operator 
$\ASP(X,Y,Z)$ and then give our algorithm as a $\mu$-calculus formula 
with the predecessor operator.

\smallskip\noindent{\bf Predecessor operator.}
Consider sets $X,Y,Z \subseteq S$ such that $Y \subseteq Z \subseteq X$; and 
we consider the following three sets of actions: 
\begin{enumerate}
\item{} $\Allow(s,X)=\set{a_1 \in \Gamma_1(s) \mid \forall a_2\in \Gamma_2(s) : \dest(s,a_1,a_2)\subseteq X}$;
\item{} $\Bad(s,X,Y)=\set{a_2 \in \Gamma_2(s) \mid \exists a_1 \in \Allow(s,X) : \dest(s,a_1,a_2)\cap Y \neq \emptyset}$; and
\item{} $\Good(s,X,Y,Z)=\set{a_1\in \Allow(s,X) \mid \forall a_2\in (\Gamma_2(s)\setminus\Bad(s,X,Y)): \dest(s,a_1,a_2)\subseteq Z\wedge \cost(s,a_1,a_2)=1}$.
\end{enumerate}
The intuitive description of the action sets are as follows: 
(i)~the set $\Allow(s,X)$ consists of all actions for player~1, such that against all 
actions of player~2 the set $X$ is not left; 
(ii)~the set $\Bad(s,X,Y)$ is the set of player~2 actions $a_2$, such that there 
is a player~1 action $a_1$ in $\Allow(s,X)$ such that given $a_1$ and $a_2$ 
the set $Y$ is reached in one-step from $s$ with positive probability; 
and 
(iii)~$\Good(s,X,Y,Z)$ is the set of actions for player~1 in $\Allow(s,X)$ such that for all
actions for player~2 that are not in $\Bad(s,X,Y)$ the next state is in $Z$
and the reward is~1. 
The set $\ASP(X,Y,Z)$ is the set of states where $\Good(s,X,Y,Z)$ is 
non-empty, i.e., $\ASP(X,Y,Z)=\set{s \mid \Good(s,X,Y,Z)\neq \emptyset}$
(the word $\ASP$ is an acronym for allow-stay-progress, i.e., (i)~it allows 
the play to remain in $X$ forever, and either (ii)~progress to $Y$ with positive probability 
or (iii)~stay in $Z$ and get reward~1 with high probability).
Let $X^*= \nu X. \mu Y. \nu Z. \ASP(X,Y,Z)$ be the fixpoint. 
We will show that 
\[
X^*=\Almost_1(\LimInfAvg(1)) = \Almost_1(\LimSupAvg(1)).
\]
Moreover we will show that 
$X^*= \bigcap_{\epsilon>0}  \Almost_1(\LimInfAvg(1-\epsilon)) 
    = \bigcap_{\epsilon>0}  \Almost_1(\LimSupAvg(1-\epsilon))$
(see Theorem~\ref{thm:almost}).
In the following two lemmas we establish that 
$X^*\subseteq \Almost_1(\LimInfAvg(1)) \subseteq  \Almost_1(\LimSupAvg(1))$ as follows:
(1)~in the first lemma we show that for all $\epsilon>0$ there 
is a stationary strategy to ensure that from all states in $X^*$ 
that the limit-inferior mean-payoff is at least $1-\epsilon$ 
and the set $X^*$ is never left; 
and (2)~in the second lemma we use the stationary strategies of the first
lemma repeatedly to construct a Markov almost-sure winning strategy.

\begin{lemma}\label{lemm:almost1}
For all $\epsilon>0$, there exists a stationary strategy 
$\stra_1^\epsilon$ with patience at most $(\frac{n \cdot m}{\delta_{\min}\cdot \epsilon})^{n^{n+2}}$, such that for all strategies $\stra_2$ and all 
$s \in X^*$ we have 
$\Pr_s^{\stra_1^\epsilon,\stra_2}(\LimInfAvg(1-\epsilon) \cap \Safe(X^*))=1$.
\end{lemma}
\begin{proof}
We first analyse the computation of $X^*$. We have 
\[
X^*= \mu Y. \nu Z. \ASP(X^*,Y,Z);
\]
this is achieved by simply replacing $X$ with $X^*$ in the $\mu$-calculus expression $\nu X. \mu Y. \nu Z. \ASP(X,Y,Z)$,
then getting rid of the outer-most $\nu$ quantifier, and evaluating the rest of the $\mu$-calculus expression.
Since $X^*$ is a fixpoint we have  $X^*= \mu Y. \nu Z. \ASP(X^*,Y,Z)$.
Thus the computation of $X^*$ is achieved as follows: 
we have $Y_0=\emptyset$ and $Y_{i+1}= \nu Z. \ASP(X^*,Y_i,Z)$. 
Let $\ell$ be the smallest number such that 
$Y_{\ell}= \nu Z. \ASP(X^*,Y_\ell,Z)$, and we have $Y_{\ell}=X^*$. 
For a state $s \in X^*$, let 
$\Aw(s)=|\Allow(s,X^*)|$ denote the size of the set of allowable actions;
and for $j\geq 0$, let $\Gd(s,j)=|\Good(s,X^*,Y_{\ell-j-1},Y_{\ell-j})|$
denote the size of the set of good actions for the triple $X^*,Y_{\ell-j-1},Y_{\ell-j}$.

Fix $\epsilon>0$. 
The desired strategy $\sigma_{1}^\epsilon$ will be constructed as a finite sequence of strategies 
$\sigma_1^{\epsilon,1},\sigma_1^{\epsilon,2},\dots\sigma_1^{\epsilon,\ell}$. 
We start with the definition of $\sigma_{1}^{\epsilon,1}$, and for a state $s$ and action $a \in \mov_1(s)$ 
we have the following:
if $s \in (Y_{\ell} \setminus Y_{\ell-1})$, then \begin{equation*}
\sigma_1^{\epsilon,1}(s)(a)=
\begin{cases}
\frac{1-\epsilon}{\Gd(s,0)}   & \text{$a\in \Good(s,X^*,Y_{\ell-1},Y_{\ell})$ } 
\text{and $\Aw(s) \neq \Gd(s,0) $}\\[2ex]
\frac{1}{\Gd(s,0)} & \text{$a\in \Good(s,X^*,Y_{\ell-1},Y_{\ell})$ } 
\text{and  $\Aw(s) = \Gd(s,0) $}\\[2ex]
\frac{ \epsilon}{\Aw(s) -\Gd(s,0)}  &  \text{$a\in (\Allow(s,X^*)\setminus \Good(s,X^*,Y_{\ell-1},Y_{\ell}))$}\\[2ex]
0 &  \text{$a\not \in \Allow(s,X^*)$} 
\end{cases}\end{equation*}
and if $s \not\in (Y_{\ell}\setminus Y_{\ell-1})$, then $\sigma_1^{\epsilon,1}(s)$ is an arbitrary probability distribution over $\mov_1(s)$.
For $j> 1$, let $\beta_{j}= n^{-\frac{n^{j-1}-1}{n-1}}\cdot (\frac{m}{\delta_{\min}})^{-\frac{n^{j}-1}{n-1}+1}\cdot \epsilon^{\frac{n^{j}-1}{n-1}}$.
We now define $\sigma_1^{\epsilon,j}$, for $j>1$.
For a state $s$ and action $a \in \mov_1(s)$, we have the following:
if $s \in (Y_{\ell-j} \setminus Y_{\ell-j-1})$,
\begin{equation*}
\sigma_1^{\epsilon,j}(s)(a)=
\begin{cases}
\frac{1 -\beta_j}{\Gd(s,j)}  & \text{$a\in \Good(s,X^*,Y_{\ell-j-1},Y_{\ell-j})$ } 
\text{and $\Aw(s)\neq \Gd(s,j) $}\\[2ex]
\frac{1}{\Gd(s,j)} & \text{$a\in \Good(s,X^*,Y_{\ell-j-1},Y_{\ell-j})$ } 
\text{and $\Aw(s) = \Gd(s,j)$}\\[2ex]
\frac{ \beta_j}{\Aw(s) -\Gd(s,j)}  &  \text{$a\in (\Allow(s,X^*)\setminus \Good(s,X^*,Y_{\ell-j-1},Y_{\ell-j}))$}\\[2ex]
0 &  \text{$a\not \in \Allow(s,X^*)$} 
\end{cases}\end{equation*}
and if $s \not\in (Y_{\ell-j}\setminus Y_{\ell-j-1})$, then 
$\sigma_1^{\epsilon,j}(s)(a) = \sigma_1^{\epsilon,j-1}(s)(a)$. 
The strategy $\sigma^\epsilon_1$ is then $\sigma^{\epsilon,\ell}_1$.

\smallskip\noindent{\em Bounds on patience.}
Observe that the patience is at most $(\frac{n\cdot m}{\delta_{\min}\cdot \epsilon})^{n^{n+2}}$, because that is a bound on the inverse of $\beta_\ell$, by definition (note that $\ell$ is at most $n$).

We will now show that $\sigma^\epsilon_1$ has the desired properties to ensure
the safety and limit-average objectives.

\smallskip\noindent{\em Ensuring safety.}
First observe that the strategy $\sigma_1^\epsilon$ never plays actions not in $\Allow(s,X^*)$, for 
states $s \in X^*$. For all actions $a_1\in \Gamma_1(s)$, if there is an action $a_2\in \Gamma_2(s)$ such that $\dest(s,a_1,a_2) \cap (S\setminus X^*)
\neq \emptyset$, then $a_1$ does not belong to $\Allow(s,X^*)$ and hence 
is played with probability~0 (at $s$ for all $s \in X^*$). This implies that for all $s' \in X^*$ and for all 
strategies $\sigma_2$ we have that 
$\Pr_{s'}^{\sigma_1^\epsilon,\sigma_2}(\Safe(X^*))=1$.
Hence the safety property is guaranteed.

\smallskip\noindent{\em Ensuring $\LimInfAvg(1-\epsilon)$.}
We now focus on the mean-payoff objective. Since the strategy $\sigma_1^\epsilon$ is a stationary 
strategy, fixing the strategy $\sigma_1^\epsilon$ for player~1, we obtain an 
MDP for player~2. 
In MDPs, there exist optimal positional strategies for the player 
to minimize mean-payoff objectives~\cite{LigLip69}. 
Hence we only focus on positional strategies as counter 
strategies for player~2 against $\sigma_1^\epsilon$.

We will show the following by induction on $j$: for all positional strategies $\sigma_2$ 
for player~2, for all $s \in (X^*\setminus Y_{\ell-j})$ one of the following two properties hold: 
either (1)~the set $Y_{\ell-j}$ is reached within at most 
$\frac{\epsilon}{\beta_{j+1}}$ steps in expectation;  
or 
(2)~we have 
\begin{equation}\label{eqt: get more than epsilon}
\textstyle{\Pr_s^{\sigma_1^{\epsilon,j},\sigma_2}}(\LimInfAvg(1-\epsilon))=1.
\end{equation}
We present the inductive proof now.

\smallskip\noindent{\em Base case.}
First the base case, $j=1$. Let $s\in (Y_{\ell}\setminus Y_{\ell-1})=(X^*\setminus Y_{\ell-1})$. 
Consider $\sigma_1^{\epsilon,1}$ and a positional strategy $\sigma_2$ for player~2.
After fixing both the chosen strategies, since both the strategies are stationary we obtain a Markov chain.
Let the random variable indicating the play from $s$ in the Markov chain be denoted as $P$.
There are now two cases. 
\begin{itemize}
\item 
We consider the case when the play $P$ enters a state $s_1$ from which no state $s_2$ can be reached, where $\sigma_2$ plays an action in 
$\Bad(s_2,X^*,Y_{\ell-1})$.
Then once $s_1$ is reached, for any state $s_3$ that appears after $s_1$ we have that 
$\sigma_1^{\epsilon,1}$ plays some action in $\Good(s_3,X^*,Y_{\ell-1},Y_{\ell})$ 
with probability $1-\epsilon$ (and hence get a payoff of~1). Hence $P$ satisfies Equation~\ref{eqt: get more than epsilon} in this case. 

\item In the other case the play $P$ can always reach a state $s_2$ such that $\sigma_2$ plays an action in 
$\Bad(s_2,Y_{\ell},Y_{\ell-1})=\Bad(s_2,X^*,Y_{\ell-1})$ 
and we can therefore enter a state in $s_3\in Y_{\ell-1}$ with probability at least $\frac{\epsilon\cdot \delta_{\min}}{m}$ from $s_2$.
Since $\frac{\epsilon\cdot \delta_{\min}}{m}$ is a lower bound on the smallest positive probability in the Markov chain, 
any state that can be reached is actually reached within at most $n$ steps with  probability at least $(\frac{\epsilon\cdot \delta_{\min}}{m})^n$. 
Hence the probability to reach a state in $Y_{\ell-1}$ within at most $n$ steps  is at least $(\frac{\epsilon\cdot \delta_{\min}}{m})^n$. 
Therefore we need at most $n\cdot (\frac{m}{\epsilon\cdot \delta_{\min}})^{n}=\frac{\epsilon}{\beta_2}$ 
steps in expectation to reach $Y_{\ell-1}$.

\end{itemize}

\smallskip\noindent{\em Inductive case.} 
We now consider the inductive case for $j>1$, and the argument is similar to the base case. 
Let $s\in (X^*\setminus Y_{\ell-j})$. As above we fix a positional strategy $\sigma_2$ for 
player~2 and consider the Markov chain induced by $\sigma_1^\epsilon$ and $\sigma_2$, and 
denote the random variable for a play from $s$ in the Markov chain as $P$.
We have two cases.
\begin{itemize}
\item 
We consider the case when the play $P$ enters a state $s_1$ from which no state $s_2$ can be reached, where $\sigma_2$ plays an action 
in $\Bad(s_1,X^*,Y_{\ell-j})$.
Hence once $s_1$ is reached along $P$, no state in $Y_{\ell-j}$ can be reached along the play.
We can view the states in $(Y_{\ell}\setminus Y_{\ell-j+1})$ as either 
(i)~already satisfying the desired Equation~\ref{eqt: get more than epsilon} by the inductive hypothesis, 
or (ii)~else entering a state in $(Y_{\ell-j+1}\setminus Y_{\ell-j})$ (no state in $Y_{\ell-j}$ can be reached) 
after having giving payoff at least~0 for at most $\frac{\epsilon}{\beta_j}$ time 
steps in expectation (by the inductive hypothesis). 
But in all states $s_3$ in $(Y_{\ell-j+1}\setminus Y_{\ell-j})$ that the play $P$ visits 
after entering the state $s_1$, the strategy $\sigma_1^{\epsilon,0}$ chooses an 
action in $\Good(s_3,X^*,Y_{\ell-j},Y_{\ell-j+1})$ with probability  
$1-\beta_j$; 
and hence from $s_3$ the play $P$ enters another state in $(Y_{\ell-j+1}\setminus Y_{\ell-j})$ 
and get payoff~1. 
If we do not get payoff~1 in any state in $(Y_{\ell-j+1}\setminus Y_{\ell-j})$, 
we expect to get at most $\frac{\epsilon}{\beta_j}$ times payoff~0 
and then again enter some state in $(Y_{\ell-j+1}\setminus Y_{\ell-j})$. 
Hence the probability that any given payoff is~0 is at most 
\[\frac{\beta_j\cdot \frac{\epsilon}{\beta_j}}{1-\beta_j+\beta_j\cdot \frac{\epsilon}{\beta_j}}\leq \epsilon\] 
and hence the play satisfies Equation~\ref{eqt: get more than epsilon}.

\item In the other case the play $P$ can always reach a state $s_2$ 
such that $\sigma_2$ plays an action in $\Bad(s_2,X^*,Y_{\ell-j})$ 
and we can therefore enter a state in $s_3\in Y_{\ell-j}$ with probability at least 
$\kappa_j=\frac{\beta_j\cdot \delta_{\min}}{m}$ from $s_2$.
Since $\kappa_j$ is a lower bound on the smallest positive probability in the Markov chain, 
any state that can be reached is actually reached within at most $n$ steps with  probability at least $\kappa_j^n$.
In expectation we will need $p^{-1}$ trials before an event that happens with probability $p>0$ happens. We therefore needs $\kappa_j^{-n}$ trials, each using $n$ steps for a total of 
\[n\cdot \kappa_j^{-n}\] steps. 
Therefore we need at most 
\[
\begin{split}
n\cdot \kappa_j^{-n} &= n\cdot\frac{m^n}{\beta_j^n\cdot \delta_{\min}^n}\\
&= n\cdot m^n \cdot \delta_{\min}^{-n}\cdot n^{n\cdot\frac{n^{j-1}-1}{n-1}}\cdot \left(\frac{m}{\delta_{\min}}\right)^{n\cdot(\frac{n^{j}-1}{n-1}-1)}\cdot \epsilon^{-n\cdot\frac{n^{j}-1}{n-1}}\\
&=\frac{\epsilon}{\beta_{j+1}}
\end{split}\] 
steps in expectation to reach $Y_{\ell-j}$.

\end{itemize}
Note that if Equation~\ref{eqt: get more than epsilon} is not satisfied and the condition~(1) 
(that the set $Y_{\ell-j}$ is reached after at most $\frac{\epsilon}{\beta_{j+1}}$ steps in expectation) is satisfied,
then it implies that $Y_{\ell-j}$ is reached eventually with probability~1.
Hence by induction it follows that either Equation~\ref{eqt: get more than epsilon} is satisfied by $\sigma^{\epsilon,\ell}_1$ 
or $Y_0$ is reached eventually with probability~1.
Since $Y_0$ is the empty set, if player~1 plays  $\sigma^{\epsilon,\ell}_1=\sigma^\epsilon_1$ and 
player~2 plays any positional strategy $\sigma_2$, then Equation~\ref{eqt: get more than epsilon}
must be satisfied, i.e., for all $s \in X^*$, for all positional strategies of player~2 we have 
$\Pr_s^{\sigma^\epsilon_1,\sigma_2}(\LimInfAvg(1-\epsilon))=1$.
But since  $\sigma^\epsilon_1$ is stationary, as already mentioned, the mean-payoff objective is minimized by a positional 
strategy for player~2. 
Hence, it follows that for all $s \in X^*$ and all strategies for player~2 (not necessarily 
positional) we have 
$\Pr_s^{\sigma^\epsilon_1,\sigma_2}(\LimInfAvg(1-\epsilon))=1$.
Since safety is already ensured by $\sigma^\epsilon_1$, it follows that for all 
$s \in X^*$ and all strategies for player~2 we have $\Pr_s^{\sigma^\epsilon_1,\sigma_2}(\LimInfAvg(1-\epsilon) \cap \Safe(X^*))=1$.
\end{proof}

\begin{lemma}\label{lemm:almost2}
Let $U$ be a set of states such that for all $\epsilon>0$  
there exists a stationary strategy $\stra_1^\epsilon$ that against all strategies 
$\stra_2$ and all $s \in U$, ensures \[\Pr\nolimits_s^{\stra_1^\epsilon,\stra_2}(\LimInfAvg(1-\epsilon) \cap \Safe(U))=1.\]
Then there exists a Markov strategy $\stra_1^*$ for player~1, such that for all strategies 
$\stra_2$ and all $s\in U$ we have \[\Pr\nolimits_s^{\stra_1^*,\stra_2}(\LimInfAvg(1))=1.\]
\end{lemma}
\begin{proof}
The construction of the desired strategy $\sigma_1^*$ is as follows: 
consider the sequence $\epsilon_1,\epsilon_2,\ldots$ such that 
$\epsilon_1=\frac{1}{4}$ and $\epsilon_{i+1}=\frac{\epsilon_i}{2}$.
At any point, the strategy $\sigma_1^*$ will play according to $\sigma_1^{\epsilon_i}$ 
for some $i \geq 1$. 
Initially the strategy plays as $\sigma_1^{\epsilon_1}$.
The strategy $\sigma_1^{\epsilon_i}$ ensures that against any strategy $\sigma_2$ 
and starting in any state $s \in U$ we get that 
$\Pr_s^{\sigma_1^{\epsilon_i},\sigma_2}(\LimInfAvg(1-\epsilon_i)) =1$. 
Hence after a finite number of steps (that can be upper bounded with a bound $J_i$) 
with probability~1, the average-payoff 
is at least $1-2\cdot \epsilon_i$ against any counter-strategy of player~2; 
and the safety objective ensures that the set $U$ is never left. 
The bound $J_i$ can be pre-computed: an easy description of the computation of the bound 
$J_i$ is through value-iteration (on the player-2 MDP obtained by fixing the stationary 
strategy $\sigma_1^{\epsilon_i}$ as the strategy for player~1), 
and playing the game for a finite number of steps that 
ensure limit-average $1-2\cdot\epsilon_i$ with probability~1, and then use the finite 
number as the bound $J_i$.
Once the payoff is at least  $1-2\cdot \epsilon_i$ the strategy switches 
to the strategy $\sigma_1^{\epsilon_{i+1}}$ for $J_{i+1}$ steps.
As the length of the play goes to $\infty$, for all $\epsilon>0$, 
for all $s \in U$ and all strategies $\sigma_2$ we have 
$\Pr_s^{\sigma_1^*,\sigma_2}(\LimInfAvg(1-\epsilon))=1$, 
and since this holds for all $\epsilon>0$, we have 
$\Pr_s^{\sigma_1^*,\sigma_2}(\LimInfAvg(1))=1$. 
Using the bounds on the sequence $(J_i)_{i\geq 1}$ for the number of 
steps required before we switch strategies in the sequence of strategies 
$(\sigma_1^{\epsilon_i})_{i\geq 1}$,
we obtain that the strategy $\sigma_1^*$ is a Markov strategy.
The desired result follows.
\end{proof}

Lemma~\ref{lemm:almost1} and Lemma~\ref{lemm:almost2} establishes one required inclusion (Lemma \ref{lemm:incl1}),
and we establish the other inclusion in Lemma~\ref{lemm:not-almost}.

\begin{lemma}\label{lemm:incl1}
We have $X^* \subseteq \Almost_1(\LimInfAvg(1))$.
\end{lemma}

\begin{lemma}\label{lemm:not-almost}
 We have
\[\ov{X}^*= (S\setminus X^*) 
\subseteq 
\set{s \in S \mid \exists \sigma_2^* \forall \sigma_1. 
\Pr\nolimits_s^{\sigma_1,\sigma_2^*}(\overline{\LimSupAvg}_{\leq}(1-c))>0},\] 
where $c= (\frac{\delta_{\min}}{m})^{n-1}\cdot \frac{1}{m}$. 
Moreover, there exist witness Markov strategies $\sigma_2^*$ for player~2 to ensure $\overline{\LimSupAvg}_{\leq}( 1-c)>0$ from 
$\ov{X}^*$.
\end{lemma}

\begin{proof}
We will construct a Markov strategy $\sigma_2$ for player~2, such that the limit supremum average reward is at most $1-c$ with positive probability for plays that start in a state in $\ov{X}^*$. This implies that we have $\ov{X}^* 
\subseteq 
\set{s \in S \mid \exists \sigma_2^* \forall \sigma_1. 
\Pr\nolimits_s^{\sigma_1,\sigma_2^*}(\overline{\LimSupAvg}_{\leq}( 1-c))>0}$.
We first consider the computation of $X^*$.
Let $X_0=S$ and $X_{i}=\mu Y. \nu Z.\ASP(X_{i-1},Y,Z)$, for $i\geq 1.$ 
Thus we will obtain a sequence $X_0 \supset X_1 \supset X_2 \cdots \supset X_{k-1} \supset X_{k} = X_{k+1} =X^*$.
For a set $U$ of states, let us denote by $\ov{U}=(S\setminus U)$ the complement of the set $U$.
We will construct a spoiling strategy $\sigma_2^*$ for player~2 as the end of a sequence of strategies, $\sigma^1_2,\sigma^2_2,\dots,\sigma^k_2=\sigma_2^*$, where $k\leq n$.
The strategy $\sigma^j_2$ will be constructed such that for all strategies $\sigma_1$ for player~1, 
for all $0\leq j \leq k$, and for all $s \in \ov{X}_j$, we have 
\begin{enumerate}
\item \emph{(Property~1).} Either $\Pr_s^{\sigma_1,\sigma_2^j}(\overline{\LimSupAvg}_{\leq} (1-c)) >0$; or 
\item \emph{(Property~2).} $\Pr_s^{\sigma_1,\sigma_2^j}(\Reach(\ov{X}_{j-1})) >0$ 
(recall that $\ov{X}_{j-1}=(S \setminus X_{j-1})$).
\end{enumerate}
The proof of the result will be by induction on $j$, and 
intuitively the correctness of the strategy construction of $\sigma_2^j$ will use 
the nested iteration of the $\mu$-calculus formula for $X^*$.

We assume that $X^*\neq S$, because if $S=X^* \subseteq \Almost_1(\LimSupAvg(1))$,
then $(S\setminus X^*)$ is the empty set and we are trivially done.

\smallskip\noindent{\em Construction of $\sigma_2^1$.} We first describe the details 
of $\sigma_2^1$ as the later strategies will be constructed similarly.
Since $X^*\neq S$, we have that $\ov{X}_1=(S\setminus X_1)$ is non-empty.
We first show that for all $s \in \ov{X}_1$ we have that 
$(\mov_2(s) \setminus \Bad(s,S,X_1))$ is non-empty;
otherwise if $(\mov_2(s) \setminus \Bad(s,S,X_1))$ is empty, then $\Good(s,S,X_1,X_1)$ is the whole 
set $\mov_1(s)$ of actions, which implies $s \in \ASP(S,X_1,X_1)=X_1$ 
(contradicting that $s \in \ov{X}_1$).
The description of the strategy $\sigma_2^1$ is as follows: for all $s \in \ov{X}_1$ 
the strategy plays all actions in $(\mov_2(s) \setminus \Bad(s,S,X_1))$ 
uniformly at random; and  for $s$ not in $\ov{X}_1$, the strategy $\sigma_2^1(s)$ is 
an arbitrary probability distribution over $\mov_2(s)$.
To prove the correctness of the construction of $\sigma_2^1$ we analyse the 
computation of the set $X_1$ as follows:
the set $X_1$ is obtained as a sequence 
$Z_1^0 \supset Z_1^1  \supset Z_1^2 \supset \cdots \supset Z_1^\ell = Z_1^{\ell+1} =X_1$
where $Z_1^0=S$ and $Z_1^{i+1} = \ASP(S,X_1,Z_1^i)$.

\begin{enumerate}
\item We first show that for all states $s$ in $\ov{Z}_1^1=(S \setminus Z_1^1)$, we have:
(1)~the next state is in the set $\ov{X}_1$ with probability~1, and 
(2)~the probability that the reward is~0 in one step from $s$ is at least $\frac{1}{m}$.
We know that the set $\Good(s,S,X_1,Z_1^0)=\Good(s,S,X_1,S)=\emptyset$.
Moreover, $\Allow(s,S)$ is the set of all player~1 actions $\mov_1(s)$. 
First, for every action $a$ of player~1, for all actions $b \in (\mov_2(s) \setminus \Bad(s,S,X_1))$ 
we have $\dest(s,a,b) \subseteq \ov{X}_1$, and hence it follows 
that the set $\ov{X}_1$ is never left (i.e., the next state is always in $X_1$ 
with probability~1).
Second, since $\Good(s,S,X_1,S)=\emptyset$, for every action $a$ of player~1,
there exists an action $b \in (\mov_2(s) \setminus \Bad(s,S,X_1))$ such that 
$\cost(s,a,b)=0$ (note that for all actions $a$ and $b$ the condition 
$\dest(s,a,b) \subseteq Z_1^0=S$ is trivially satisfied, and hence the reward 
must be~0 to show that the action does not belong to the good set of actions).
Since all actions in  $(\mov_2(s) \setminus \Bad(s,S,X_1))$ are played uniformly 
at random for every action $a$ for player~1 the probability that the reward
is~0 in one step is at least $\frac{1}{m}$.

\item For $i >1$ we show that for all states $s$ in $\ov{Z}_1^i=(S \setminus Z_1^i)$, we have:
(1)~the next state is in the set $\ov{X}_1$ with probability~1; and 
(2)~either (i)~the probability to reach the set $\ov{Z}_1^{i-1}$ in one step from $s$ is at least 
$\frac{\delta_{\min}}{m}$ or (ii)~the probability to get reward~0 in one step from $s$ is at least $\frac{1}{m}$.
As in the previous case since $\Allow(s,S)$ is the set of all actions $\mov_1(s)$,
it follows from the same argument as above that the next state is in $\ov{X}_1$ 
with probability~1.
We now focus on the second part of the claim.
We know that $\Good(s,S,X_1,Z_1^{i-1})$ is empty. 
Hence for all actions $a$ for player~1 
there exists an action $b \in (\mov_2(s) \setminus \Bad(s,S,X_1))$
such that either (i)~$\cost(s,a,b)=0$, or (ii)~$\dest(s,a,b) \cap \ov{Z}_1^{i-1} 
\neq \emptyset$ (i.e., $\dest(s,a,b) \not\subseteq Z_1^{i-1}$).
It follows that either the reward is~0 with probability at least $\frac{1}{m}$ 
in one step from $s$ or the set $\ov{Z}_1^{i-1}$ is reached with probability at 
least $\frac{\delta_{\min}}{m}$ in one step from $s$.

\end{enumerate}

It follows from above that from any state in $(S\setminus X_1)$, there is a path of length at 
most $n$ such that each step occurs with probability atleast $\frac{\delta_{\min}}{m}$ (except for the last which occurs with probability at least $\frac{1}{m}$)
and the last reward is~0, and the path always stays in $(S\setminus X_1)$,
given player~2 plays the strategy $\sigma^1_2$, irrespective of the strategy of player~1.
Hence for all $s \in \ov{X}_1=(S \setminus X_1)$ and for all strategies 
$\sigma_1$ we have $\Pr_s^{\sigma_1,\sigma_2^1}(\overline{\LimSupAvg}_{\leq}( 1-c))=1$.
We present a remark about the above construction as it will be used later.

\begin{remark}\label{rem:proof}
Let $X_1=\mu Y. \nu Z. \ASP(S,Y,Z)$, and $\ov{X}_1=(S\setminus X_1)$.
Then there exists a stationary strategy $\sigma_2^1$ with patience at most $m$ for 
player~2 such that for all strategies $\sigma_1$ for player~1 we have 
$\Pr_s^{\sigma_1,\sigma_2^1}(\overline{\LimSupAvg}_{\leq}( 1-c))=1$,
for all $s \in \ov{X}_1$.
\end{remark}

We now describe the inductive construction of the strategy $\sigma_2^j$ from $\sigma_2^{j-1}$, for $j\geq 2$. 
Let $0<\epsilon<1$ be given.
For plays which are in state $s\not \in (X_{j-1}\setminus X_{j})$, the strategy $\sigma_2^{j}$ follows $\sigma_2^{j-1}$. 
If the play is in state $s\in (X_{j-1}\setminus X_{j})$ in round $i$ the strategy $\sigma_2^{j}$ uses a binary random variable $B^i$ (where $B^i$ is independent of $B^\ell$ for $\ell\neq i$) 
which is~1 with probability  $\frac{\epsilon}{2^i}$ and~0 otherwise. 
If $B^i$ is~1, then $\sigma_2^{j}$ chooses an action uniformly at random from $\Gamma_2(s)$, otherwise it chooses an action uniformly at random from 
$(\Gamma_2(s)\setminus\Bad(s,X_{j-1},X_j))$. Notice the fact that $B^i$ is independent of 
$B^\ell$, for $\ell\neq i$, ensures that $\sigma_2^j$ is a Markov strategy. 
We show that $(\Gamma_2(s)\setminus\Bad(s,X_{j-1},X_j))$ is non-empty. 
If $\Allow(s,X_{j-1})$ is empty, then $\Bad(s,X_{j-1},X_j)$ is empty and 
therefore $(\Gamma_2(s)\setminus\Bad(s,X_{j-1},X_j))$ is non-empty. 
Otherwise, if $\Allow(s,X_{j-1})$ is non-empty, 
we can use that $\Good(s,X_{j-1},X_j,X_j)$ is empty, 
because $s\in (X_{j-1}\setminus X_{j})$. 
But by definition of $\Good(s,X_{j-1},X_j,X_j)$ this implies that $(\Gamma_2(s)\setminus\Bad(s,X_{j-1},X_j))$ is non-empty. 

Consider a counter-strategy $\sigma_1$ for player~1.
In round $\ell$, if the play is in a state $s$ in $(X_{j-1}\setminus X_j)$ and the conditional probability 
that  $\sigma_1$ chooses an action $a$ with positive probability which is not in $\Allow(s,X_{j-1})$, 
then there is an action $b$ such that $\dest(s,a,b)\cap \ov{X}_{j-1} \neq \emptyset$. 
Hence since $\sigma_2^j$ plays all actions with positive probability we see that such plays reaches 
$(S\setminus X_{j-1})=\ov{X}_{j-1}$ with positive probability (in this case the 
desired Property~2 holds). 
Therefore, we consider the case such that $\sigma_1$ only plays action with positive probability that are in 
$\Allow(s,X_{j-1})$. 
With probability at least $1-\sum_{i=1}^\infty \frac{\epsilon}{2^i}=1-\epsilon>0$, we have that $B^i=0$ for all $i$. 
If $B^i$ is~0 for all $i$, the proof proceeds like in the base case (correctness proof for $\sigma_2^1$), 
except that we view $X_{j-1}$ as the set of all states (note that no state outside $X_{j-1}$ can be reached 
because $\sigma_1$ only plays actions in $\Allow(s,X_{j-1})$ and that  $B^i=0$ for all $i$).
In this scenario, as in the proof of the base case, we have that the strategy $\sigma_2^j$ 
ensures that all plays starting in states $s \in (\ov{X}_j \setminus \ov{X}_{j-1})$ do not 
leave the set $(\ov{X}_j \setminus \ov{X}_{j-1})$, and the probability that the objective $\overline{\LimSupAvg}_{\leq}( 1-c)$ 
is satisfied is strictly greater than~0 for all strategies of player~1.
This establishes by induction the desired Properties~1 and~2.
Since $\ov{X}_0$ is empty, it follows that for all states $s \in \ov{X}^*$ and 
any strategy $\sigma_1$ for player~1, we have 
$\Pr_s^{\sigma_1,\sigma_2^*}(\overline{\LimSupAvg}_{\leq}(1-c)) >0$ 
(as in the proof of Lemma~\ref{lemm:almost1}). 
Notice that since $\sigma_2^j$  is a Markov strategy for all $j$, it follows that $\sigma_2^*$ is also a Markov strategy.
The desired result is established.
\end{proof}

\begin{theorem}[Qualitative determinacy and polynomial-time computability]\label{thm:almost}
The following assertions hold for all concurrent game structures with boolean rewards:
\begin{enumerate}

\item We have 
\[
\begin{array}{rcl}
X^* & = & \Almost_1(\LimInfAvg(1)) = \Almost_1(\LimSupAvg(1)) \\
& = & 
\displaystyle 
\bigcap_{\varepsilon >0} \Almost_1(\LimInfAvg(1-\varepsilon)) =
\bigcap_{\varepsilon >0} \Almost_1(\LimSupAvg(1-\varepsilon));
\end{array}
\] 
and 
\[
\begin{array}{rcl}
(S\setminus X^*) & = & 
\Positive_2(\ov{\LimInfAvg}(1)) = \Positive_2(\ov{\LimSupAvg}(1)) \\
& = & 
\displaystyle
\bigcup_{c>0} \Positive_2(\ov{\LimInfAvg}_{\leq}(1-c)) 
=\bigcup_{c>0} \Positive_2(\ov{\LimSupAvg}_{\leq}(1-c));
\end{array}
\]
where $X^*=\nu X. \mu Y. \nu Z. \ASP(X,Y,Z)$.

\item The set $X^*$ can be computed in cubic time 
(in time $O(n^2 \cdot |\delta|)$, where $|\delta|=\sum_{s\in S}\sum_{a\in \Gamma_1(s)}\sum_{b\in \Gamma_2(s)}\left| \dest(s,a,b)\right|$)
by straight-forward  computation of the $\mu$-calculus formula
$\nu X. \mu Y. \nu Z. \ASP(X,Y,Z)$.

\end{enumerate}
\end{theorem}
\begin{proof}
Trivially we have 
$\Almost_1(\LimInfAvg(1)) \subseteq \bigcap_{\varepsilon >0} \Almost_1(\LimInfAvg(1-\varepsilon))$ and
$\bigcup_{c>0} \Positive_2(\ov{\LimInfAvg}(1-c)) \subseteq \Positive_2(\ov{\LimInfAvg}(1)) $
(also similarly for $\LimSupAvg$).
By Lemma~\ref{lemm:incl1} we have $X^* \subseteq \Almost_1(\LimInfAvg(1))$ and 
by Lemma~\ref{lemm:not-almost} we have $(S\setminus X^*) \subseteq  \bigcup_{c>0} \Positive_2(\ov{\LimSupAvg}_{\leq}(1-c))$.
Also observe that trivially we have 
$\bigcup_{c>0} \Positive_2(\ov{\LimSupAvg}_{\leq}(1-c))=
\bigcup_{c>0} \Positive_2(\ov{\LimSupAvg}(1-c))$
(and similarly for $\LimInfAvg$).
Thus we obtain all the desired equalities.
The second item trivially follows as the $\mu$-calculus formula defines a nested 
iterative algorithm.
\end{proof}

\subsection{Strategy Complexity}\label{subsec2}
In this section we will establish the complexities of the witness 
almost-sure and positive winning strategies for player~1 and player~2,
from their respective winning sets.
We start with a lemma that shows a lower bound on the time-dependent 
memory of infinite-memory strategies.
The authors would like to thank Kristoffer Arnsfelt Hansen for the proof of the following lemma.

\begin{lemma}\label{lemm:log T}
In a repeated game with absorbing states, if a strategy $\stra$ requires infinite memory, 
then more than $T$ memory states are required by the strategy for the first $T$ rounds, 
for all $T>0$, i.e., $\timedep_{\stra}(T) \geq T$,  for all $T>0$.
\end{lemma}
\begin{proof}
Let $\sigma$ be a strategy that requires infinite memory. 
Consider the directed graph where the states are the memory states of $\sigma$ and where there is an edge from state $\mem$ to state $\mem'$, 
if there exists an action $a$ consistent with $\sigma$ and an action $b$ for the other player, such that $\sigma^u(a,b,\mem)=\mem'$. Since the set of actions for each player is finite, the out-degree of all states are finite.

We have by definition of $\sigma$ that the graph is infinite  and we can reach infinitely many memory states from the start state. In a graph where each state has finite out-degree there are two possibilities. 
Either it is possible to reach a state from the start state in $T$ steps that is not reachable in $T-1$ steps, for each $T>0$; 
or only a finite number of states can be reached from the start state. Since we can reach an infinite number of states from the start state we must be in the first case. 
Therefore at least $T$ memory states can be reached from the start state in $T$ steps, for all $T>0$.
\end{proof}

Recall that a Markov strategy is an infinite-memory strategy with 
time-dependent memory of size $T$, for all $T>0$.
In view of Lemma~\ref{lemm:log T} it follows that if infinite-memory 
requirement is established for repeated games with absorbing states, 
then time-dependent memory bound of Markov strategies match the lower bound 
of the time-dependent memory.

\smallskip\noindent{\bf Infinite-memory for almost-sure winning strategies.}
In case of concurrent reachability games, stationary almost-sure winning strategies
exist.
In contrast we show that for concurrent games with boolean reward functions,
almost-sure winning strategies for exact qualitative constraint require infinite memory.

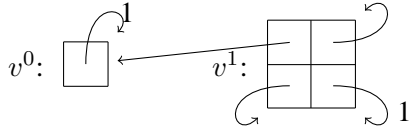
\begin{figure}
\centering
\begin{tikzpicture}[node distance=3cm]
\matrix (1) [label=left:$v^0$:,minimum height=1.5em,minimum width=1.5em,matrix of math nodes,nodes in empty cells, left delimiter={.},right delimiter={.}]
{
\\
};
\draw[black] (1-1-1.north west) -- (1-1-1.north east);
\draw[black] (1-1-1.south west) -- (1-1-1.south east);
\draw[black] (1-1-1.north west) -- (1-1-1.south west);
\draw[black] (1-1-1.north east) -- (1-1-1.south east);
\matrix (v1) [label=left:$v^1$:,right of=1,minimum height=1.5em,minimum width=1.5em,matrix of math nodes,nodes in empty cells, left delimiter={.},right delimiter={.}]
{
&\\
&\\
};
\draw[black] (v1-1-1.north west) -- (v1-1-2.north east);
\draw[black] (v1-1-1.south west) -- (v1-1-2.south east);
\draw[black] (v1-2-1.south west) -- (v1-2-2.south east);
\draw[black] (v1-1-1.north west) -- (v1-2-1.south west);
\draw[black] (v1-1-2.north west) -- (v1-2-2.south west);
\draw[black] (v1-1-2.north east) -- (v1-2-2.south east);
\draw[->](1-1-1.center) .. controls +(90:1) and +(45:1) ..  node[midway,above,right] (x) {1} (1);
\draw[->](v1-1-1.center) to  (1);
\draw[->](v1-1-2.center) .. controls +(0:1) and +(45:1.5) ..   (v1);
\draw[->](v1-2-1.center) .. controls +(180:1) and +(-135:1.5) ..   (v1);
\draw[->](v1-2-2.center) .. controls +(0:1) and +(-45:1.5) ..  node[midway,below,right] (x) {1} (v1);
\end{tikzpicture}

\caption{The example illustrates $G^1$ where all states are in $\Almost_1(\LimInfAvg(1))$, but no finite-memory almost-sure winning strategy 
exists for player~1 for the objective $\LimInfAvg(1)$.
All transitions with reward different from~0 (i.e., reward~1) have the reward annotated on the transition. \label{fig:G1}}
\end{figure}

\begin{figure}
\centering
\begin{tikzpicture}[node distance=3cm]
\matrix (v0) [label=left:$v^0$:,minimum height=1.5em,minimum width=1.5em,matrix of math nodes,nodes in empty cells, left delimiter={.},right delimiter={.}]
{
\\
};
\draw[black] (v0-1-1.north west) -- (v0-1-1.north east);
\draw[black] (v0-1-1.south west) -- (v0-1-1.south east);
\draw[black] (v0-1-1.north west) -- (v0-1-1.south west);
\draw[black] (v0-1-1.north east) -- (v0-1-1.south east);
\matrix (v1) [label=left:$v^1$:,right of=v0,minimum height=1.5em,minimum width=1.5em,matrix of math nodes,nodes in empty cells, left delimiter={.},right delimiter={.}]
{
&\\
&\\
};
\draw[black] (v1-1-1.north west) -- (v1-1-2.north east);
\draw[black] (v1-1-1.south west) -- (v1-1-2.south east);
\draw[black] (v1-2-1.south west) -- (v1-2-2.south east);
\draw[black] (v1-1-1.north west) -- (v1-2-1.south west);
\draw[black] (v1-1-2.north west) -- (v1-2-2.south west);
\draw[black] (v1-1-2.north east) -- (v1-2-2.south east);
\matrix (v2) [label=left:$v^2$:,right of=v1,minimum height=1.5em,minimum width=1.5em,matrix of math nodes,nodes in empty cells, left delimiter={.},right delimiter={.}]
{
&\\
&\\
};
\draw[black] (v2-1-1.north west) -- (v2-1-2.north east);
\draw[black] (v2-1-1.south west) -- (v2-1-2.south east);
\draw[black] (v2-2-1.south west) -- (v2-2-2.south east);
\draw[black] (v2-1-1.north west) -- (v2-2-1.south west);
\draw[black] (v2-1-2.north west) -- (v2-2-2.south west);
\draw[black] (v2-1-2.north east) -- (v2-2-2.south east);
\matrix (v3) [label=left:$v^3$:,right of=v2,minimum height=1.5em,minimum width=1.5em,matrix of math nodes,nodes in empty cells, left delimiter={.},right delimiter={.}]
{
&\\
&\\
};
\draw[black] (v3-1-1.north west) -- (v3-1-2.north east);
\draw[black] (v3-1-1.south west) -- (v3-1-2.south east);
\draw[black] (v3-2-1.south west) -- (v3-2-2.south east);
\draw[black] (v3-1-1.north west) -- (v3-2-1.south west);
\draw[black] (v3-1-2.north west) -- (v3-2-2.south west);
\draw[black] (v3-1-2.north east) -- (v3-2-2.south east);
\draw[->](v0-1-1.center) .. controls +(90:1) and +(45:1) ..  node[midway,above,right] (x) {1} (v0);
\draw[->](v1-1-1.center) to  (v0);
\draw[->](v1-1-2.center) to[bend left=75]  (v3);
\draw[->](v1-2-1.center) to[bend right=75]  (v3);
\draw[->](v1-2-2.center) .. controls +(0:1) and +(-45:1.5) ..  node[midway,below,right] (x) {1} (v1);
\draw[->](v2-1-1.center) to  (v1);
\draw[->](v2-1-2.center) to[bend left=75]  (v3);
\draw[->](v2-2-1.center) to[bend right=75]  (v3);
\draw[->](v2-2-2.center) .. controls +(0:1) and +(-45:1.5) ..  node[midway,below,right] (x) {1} (v2);
\draw[->](v3-1-1.center) to  (v2);
\draw[->](v3-1-2.center) .. controls +(0:1) and +(45:1.5) ..   (v3);
\draw[->](v3-2-1.center) .. controls +(180:1) and +(-135:1.5) ..   (v3);
\draw[->](v3-2-2.center) .. controls +(0:1) and +(-45:1.5) ..  node[midway,below,right] (x) {1} (v3);
\end{tikzpicture}

\caption{The example illustrates $G^3$ of the family $\{G^n\}$ where all states are in
$\bigcap_{\epsilon>0}\Almost_1(\LimInfAvg(1-\epsilon))$ but all witness stationary strategies that ensure 
so for player~1 require patience at least double exponential in $n$. 
All transitions with reward different from~0 (i.e., reward~1) have the reward annotated on the transition. \label{fig:G3}}
\end{figure}
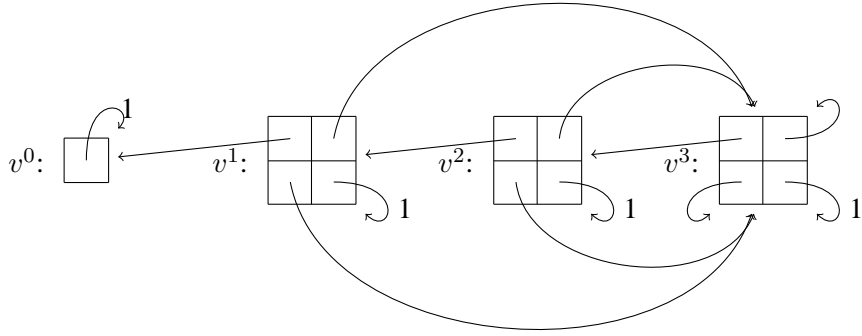

\smallskip\noindent{\em Game family.}
Let $G^n$ be the following game. The game $G^n$ has $n+1$ states, namely, $v^0$, $v^1$, $\dots$ $v^n$. The state $v^0$ is absorbing with reward~1.
For $\ell \geq 1$, the state $v^\ell$ has two actions for both players. The actions are $a^\ell_1$ and $a^\ell_2$ for player~1 and $b^\ell_1$ and $b^\ell_2$ for player~2,
respectively. 
Also $\cost(v^\ell,a^\ell_i,b^\ell_j)=0$ except for $i=j=2$, for which  $\cost(v^\ell,a^\ell_2,b^\ell_2)=1$. 
Furthermore, (i)~$\delta(v^\ell,a^\ell_i,b^\ell_j)=v^n$ for $i\neq j$; (ii)~$\delta(v^\ell,a^\ell_1,b^\ell_1)=v^{\ell-1}$; and (iii)~$\delta(v^\ell,a^\ell_2,b^\ell_2)=v^{\ell}$. 
There is an illustration of $G^1$ in Figure \ref{fig:G1} and an illustration of $G^3$ in Figure \ref{fig:G3}. 
We first show that all states are in $X^*$: in $G^n$, if we consider $X^*$ to be the set of all states
and evaluate $\mu Y. \nu Z. \ASP(X^*,Y,Z)$, then we obtain that 
$Y_0=\emptyset$, and for $i\geq 0$ we have $Y_{i+1}=\set{v^0,v^1,\ldots,v^{i}}$
because $\Allow(v^j,X^*)$ is the set of all actions available for player~1 at state $v^j$, for all $1\leq j \leq n$, 
and $\Bad(v^{i},X^*,Y_{i})=b^{i}_1$ and $\Good(v^{i},X^*,Y_{i},Y_{i+1})=a^{i}_2$.
Thus it follows that $X^*$ is the set of all states, i.e., all states belong to $\Almost_1(\LimInfAvg(1))$.
In this specific example, a Markov strategy that in round~$j\geq 0$,
for $2^{j^2}$-steps plays $a_1^1$ with probability $\frac{1}{2^{j}}$ and 
$a_2^1$ with probability $1-\frac{1}{2^{j}}$, and then goes to round $j+1$, 
is an almost-sure winning strategy for the objective $\LimInfAvg(1)$.
Note that the strategy construction described in Lemma~\ref{lemm:almost1}
and Lemma~\ref{lemm:almost2} would yield a different Markov strategy as a 
witness almost-sure winning strategy.

\begin{lemma}\label{lemm:almost super constant}
All almost-sure winning strategies for player~1 in the game $G^1$ require infinite memory for the objective $\LimInfAvg(1)$.
\end{lemma}
\begin{proof}
The proof will be by contradiction. Assume towards contradiction that there is a strategy $\sigma_1$ that uses only a finite number of memory states and is almost-sure winning for the
objective $\LimInfAvg(1)$.  
Let the smallest non-zero probability the strategy $\sigma_1$ plays $a^1_1$ in any memory state be $p$. 
We will show that there exists a strategy $\sigma_2$ for player~2 that ensures 
\[
\Pr\nolimits_{v^1}^{\sigma_1,\sigma_2}(\ov{\LimSupAvg}_{\leq}(1-p))=1.
\]
The strategy $\sigma_2$ for player~2 is to play $b^1_1$ (in $v^1$) if given the play so far, the strategy $\sigma_1$ is in a memory state where $a^1_2$ is played with probability~1. 
Otherwise player~2 plays $b^1_2$ (in $v^1$). Hence, the probability to reach $v^0$ from $v^1$ is~0. But the probability that $a^1_2$ is played at the same time as $b^1_1$ in $v^1$ 
is then at most $1-p$ in any round.
Thus we have 
$\Pr\nolimits_{v^1}^{\sigma_1,\sigma_2}(\ov{\LimSupAvg}_{\leq}(1-p))=1$ contradicting that $\sigma_1$ is an 
almost-sure winning strategy for the objective $\LimInfAvg(1)$.
It follows that every almost-sure winning strategy for player~1 requires infinite memory for the objective $\LimInfAvg(1)$ 
(note that since all states are in $X^*$ in $G^1$ it follows that almost-sure winning strategies exist for player~1).
\end{proof}

\smallskip\noindent{\bf Double exponential lower bound for patience.}
We have already established in the previous section (Lemma~\ref{lemm:almost1} and Theorem~\ref{thm:almost}) that for all $\epsilon>0$
stationary almost-sure winning strategies exist with at most double exponential patience 
for objectives $\LimInfAvg(1-\epsilon)$, for all states in $X^*$.
We now establish a double exponential lower bound on patience.

\begin{lemma}\label{lemm:almost patience lower bound}
Let $n$ be given. Given $0<\epsilon\leq \frac{1}{3}$, let $\sigma_1$ be a stationary strategy  
for player~1 that achieves 
\begin{eqnarray}\label{eq:better than 1-epsilon}
\forall v\in X^* \forall \sigma_2:\Pr\nolimits_{v}^{\sigma_1,\sigma_2}(\LimInfAvg(1-\epsilon))=1,
\end{eqnarray} 
in $G^n$. Then  $\sigma_1$ has patience at least $\epsilon^{-1.5^{n-1}}$.  
\end{lemma}
\begin{proof}
Let $n$ be given. Let $\sigma_1$ be any stationary strategy that satisfies the condition of the lemma 
(i.e., Equation~\ref{eq:better than 1-epsilon}).
Let $x_i=\sigma_1(v^{n-i})(a^{n-i}_1)$. 
First notice that $x_i>0$, otherwise, consider a stationary strategy $\sigma_2$ for player~2 such that $\sigma_2(v^{n-i})(b^{n-i}_1)=1$, 
which ensures that all payoffs of any play starting in $v^{n-i}$ would be~0.
We will now show that $x_i\leq \epsilon^{1.5^i}$ for $i<n$. The proof will be by induction on $i$. 
The proof will use two base cases $i=0$ and $i=1$, because the inductive proof then becomes simpler.

\smallskip\noindent{\em First base case.}  First the base case $i=0$. We have that $x_0\leq \epsilon$, because if $\sigma_2$ is a stationary strategy such that $\sigma_2(v^n)(b^n_2)=1$, 
then $\sigma_1(v^n)(a^n_2)\geq 1- \epsilon$ because it must satisfy the Equation~\ref{eq:better than 1-epsilon}. 
Since Equation~\ref{eq:better than 1-epsilon} is satisfied for all $\sigma_2$ we have that $0<x_0\leq \epsilon$ as desired.

\smallskip\noindent{\em Second base case.}  The second base case is for $i=1$. Let $P$ be a play starting in $v^{n-1}$. 
If $\sigma_2$ is a stationary strategy such that $\sigma_2(v^{n-1})(b^n_2)=1$ and $\sigma_2(v^{n})(b^n_1)=1$, 
then any time there is a reward of~1, the play must be in state $v^{n-1}$. But whenever $v^n$ is reached we expect at least $\epsilon^{-1}$ time steps with reward~0, before the play reaches $v^{n-1}$. 
Therefore $x_1$ must be such that $\frac{\epsilon^{-1}\cdot x_1}{1-x_1+\epsilon^{-1}\cdot x_1}\leq \epsilon$ because it must satisfy the Equation~\ref{eq:better than 1-epsilon}. 
Hence, we have that\[\begin{split}
\frac{\epsilon^{-1}\cdot x_1}{1-x_1+\epsilon^{-1}\cdot x_1}\leq \epsilon \Rightarrow\\
\frac{\epsilon^{-1}}{x_1^{-1}-1+\epsilon^{-1}}\leq \epsilon \Rightarrow\\
\epsilon^{-1}\leq \epsilon\cdot (x_1^{-1}-1+\epsilon^{-1}) \Rightarrow\\
\epsilon^{-2}\leq x_1^{-1}-1+\epsilon^{-1} \Rightarrow\\
\epsilon^{-1.5}\leq x_1^{-1} \enspace ,
\end{split}
\]
where the last implication is because $\epsilon\leq \frac{1}{3}$. 
Since Equation \ref{eq:better than 1-epsilon} is satisfied for all $\sigma_2$ we have that $0<x_1\leq \epsilon^{1.5}$ as desired.

\smallskip\noindent{\em Inductive case.}  We now consider the inductive case for $i>1$. The proof is similar to the base cases, especially the second. If $\sigma_2$ is a stationary strategy such that both $\sigma_2(v^n)(b^{n-i}_2)=1$ and $\sigma_2(v^n)(b^{n-j}_1)=1$ for $j<i$, then for any play starting in $v^{n-i}$ can only get a reward of 1 in $v^{n-i}$. But by induction $\prod_{j=0}^{i-1}x_j\leq \prod_{j=0}^{i-1}\epsilon^{1.5^j}=\epsilon^{\sum_{j=0}^{i-1} 1.5^j}=\alpha_i$. This implies that more than $\alpha_i^{-1}$ steps are needed to reach $v^{n-i}$ from $v^n$ (because clearly the play must pass through state $v^{n-j}$ for $j\leq i$). Hence whenever the play is in $v^{n-i}$, there is a reward of 1 with probability $1-x_i$ and a reward of $0$ for more than $\alpha_i^{-1}$ time steps with probability $x_i$. Hence $x_i$ must be such that \[\begin{split}
\frac{\alpha_i^{-1}\cdot x_i}{1-x_i+\alpha_i^{-1}\cdot x_i}& \leq \epsilon \Rightarrow\\
\frac{\alpha_i^{-1}}{x_i^{-1}-1+\alpha_i^{-1}}& \leq \epsilon \Rightarrow\\
\alpha_i^{-1}& \leq \epsilon\cdot (x_i^{-1}-1+\alpha_i^{-1}) \Rightarrow\\
\frac{\alpha_i^{-1}}{\epsilon}& \leq x_i^{-1}-1+\alpha_i^{-1} \Rightarrow\\
\frac{\alpha_i^{-1}}{\epsilon}+1-\alpha_i^{-1}& \leq x_i^{-1} \Rightarrow\\
\left(\frac{1}{\epsilon}-1\right)\cdot \alpha_i^{-1}& \leq x_i^{-1}\Rightarrow \\
 \alpha_i^{-1}& \leq x_i^{-1} \enspace ,\\
\end{split}
\]
where the last implication comes from the fact that $\epsilon\leq\frac{1}{3}\leq \frac{1}{2}$. 
But since $\sum_{j=0}^{i-1}1.5^j>1.5^i$ for $i>1$, the result follows.
\end{proof}

\begin{theorem}[Strategy complexity]
For concurrent games with boolean reward functions the 
following assertions hold:
\begin{enumerate}

\item Almost-sure winning strategies for objectives $\LimInfAvg(1)$ 
(and $\LimSupAvg(1)$) for player~1 require infinite memory in general;
whenever there exists an almost-sure winning strategy for objectives 
$\LimInfAvg(1)$ (and $\LimSupAvg(1)$), then a Markov almost-sure winning 
strategy exists; and the optimal bound for time-dependent memory is $T$, 
for all rounds $T>0$.

\item For all $\epsilon>0$, stationary almost-sure winning strategies 
exist for player~1 for objectives $\LimInfAvg(1-\epsilon)$ (and $\LimSupAvg(1-\epsilon)$);
and the asymptotically optimal bound for patience for such stationary 
almost-sure winning strategies is double exponential in the size of the state space.

\item Positive winning strategies for player~2 for objectives 
$\ov{\LimInfAvg}(1)$ and $\ov{\LimInfAvg}_{\leq}(1-c)$, for some constant $c>0$, 
(also $\ov{\LimSupAvg}(1)$ and $\ov{\LimSupAvg}_{\leq}(1-c)$, for some constant $c>0$)
require infinite-memory in general; whenever such positive winning strategies exist, 
Markov strategies are sufficient and the optimal bound for time-dependent memory is $T$, 
for all rounds $T>0$.

\end{enumerate}
\end{theorem}

\begin{proof}
The proofs are as follows:
\begin{enumerate}

\item Lemma~\ref{lemm:almost super constant} shows that infinite-memory is required, and 
Lemma~\ref{lemm:almost1}, Lemma~\ref{lemm:almost2}, and Theorem~\ref{thm:almost} 
show that Markov strategies are sufficient for almost-sure winning.
The sufficiency of Markov strategies establishes the $T$ upper bound 
for time-dependent memory; and Lemma~\ref{lemm:almost super constant} 
(along with the fact that the game in the lemma is a repeated game 
with absorbing states) and Lemma~\ref{lemm:log T} establishes
the $T$ lower bound for time-dependent memory.

\item The existence of stationary almost-sure winning strategies with double 
exponential patience for objectives 
$\LimInfAvg(1-\epsilon)$, for all $\epsilon>0$ follows from Lemma~\ref{lemm:almost1} and Theorem~\ref{thm:almost}.
The double exponential lower bound for patience follows from Lemma~\ref{lemm:almost patience lower bound}.

\item For the special case of concurrent reachability and safety games,
$\ov{\LimInfAvg}(1)$ and $\ov{\LimInfAvg}_{\leq}(1-c)$, for some constant $c>0$,
coincide, and the infinite-memory requirement for player~2 for positive 
winning strategies follows from~\cite{dAHK98}.
Moreover the example to show the infinite-memory requirement (from~\cite{dAHK98}) 
is a repeated game with absorbing states.
The sufficiency of Markov strategies follows from Lemma~\ref{lemm:log T};
and the optimal time-dependent memory bound of $T$ follows from 
the sufficiency of Markov strategies (upper bound) and 
Lemma~\ref{lemm:log T} and the infinite-memory requirement (lower bound).

\end{enumerate}
The desired result follows.
\end{proof}

\subsection{Improved Algorithm}\label{subsec3}
In this section we will present an improved algorithm for the 
computation of the almost-sure winning set $\Almost_1(\LimInfAvg(1))$.
The naive computation using the $\mu$-calculus formula gives a cubic
time complexity, and we will present an alternative quadratic time 
algorithm.
The key idea is to generalize the small-progress measure algorithm of~\cite{Jur00} 
with the more involved predecessor operator.

\smallskip\noindent{\bf The key intuition.}
The key intuition of the algorithm is to assign to each state $s$ a {\em level}, denoted $\ell(s)$, 
which range in the set $\{0,1,\dots,n\}$.
The level is like a ranking function and the algorithm iteratively updates the level of every state.
The initial level of each state is $n$, and the level of any state can only decrease during the execution of the algorithm. 
At the end of the execution of the algorithm, the set $X^*$ will be exactly the set of states which  have a 
strictly positive level.
The total change of levels is at most quadratic and by charging the work done to the change of the levels 
we show that the work done is also at most quadratic.

\smallskip\noindent{\bf Basic procedures.}
The algorithm will consist of two procedures, namely,  $\Preprocessing(s)$ and $\Remove(s,b)$, 
for $s\in S$ and $b\in \Gamma_2(s)$.  
To describe the procedures we first define three action sets as follows:
$\Allow(s)\subseteq \Gamma_1(s)$, $\Bad(s)\subseteq \Gamma_2(s)$ and $\Good(s)\subseteq \Gamma_1(s)$. 
The sets will have similar intuitive meaning as the corresponding set in the $\mu$-calculus expression. 
In the algorithm, whenever the set $\Good(s)$ becomes empty, the level $\ell(s)$ of $s$ will be decreased 
by one. 
For a fixed level of all the states, the sets are as follows: 
\begin{itemize}
\item $\Allow(s)$ is the set of all actions $a\in \Gamma_1(s)$ such that for all actions $b\in \Gamma_2(s)$ we have $\dest(s,a,b) \cap L_0 =\emptyset$, 
where $L_0$ is the set of states with level~0.
\item $\Bad(s)$ is the set of all actions $b\in \mov_2(s)$ such that there exists $a\in \Allow(s)$ and $t\in S$ such that $t\in \dest(s,a,b)$ and $\ell(t)>\ell(s)$.
\item $\Good(s)$ is the set of all actions $a \in \Allow(s)$ such that for all $b\in (\Gamma_2(s)\setminus \Bad(s))$ we have $\cost(s,a,b)=1$ and for all $t\in \dest(s,a,b)$ we have 
$\ell(t)\geq \ell(s)$.
\end{itemize}
For all $b\in \Gamma_2(s)$, the algorithm keeps track of the number of actions $a$ in $\Allow(s)$ and $t$ in $S$, 
such that $t\in \dest(s,a,b)$ and $\ell(t)>\ell(s)$. 
We denote this number by $\bb(s,b)$. 
Observe that an action $b\in \Gamma_2(s)$ is in $\Bad(s)$ if and only if $\bb(s,b)>0$.
We are now ready to describe the two basic procedures.
\begin{enumerate}

\item The procedure $\Preprocessing(s)$ recalculates the actions in $\Allow(s)$, $\Bad(s)$ and $\Good(s)$, 
based on the current level of all states. 
It also recalculates $\bb(s,b)$. 
The running time of the procedure is $O(\sum_{a\in \Gamma_1(s)}\sum_{b\in \Gamma_2(s)}\left| \dest(s,a,b)\right|)$,
by simple enumeration over the actions of both players and the possible successor given the actions. 
The procedure $\Preprocessing(s)$ will run (i)~once for each time state $s$ changes level;
(ii)~each time some state changes to level 0; 
and (iii)~once during the initialization of the algorithm.

\item The procedure $\Remove(s,b)$ is run only when $\bb(s,b)$ is zero.
The procedure $\Remove(s,b)$ removes $b$ from $\Bad(s)$ and for each action $a\in \Good(s)$  
checks if $\cost(s,a,b)=0$. If so, it removes such $a$'s from $\Good(s)$. 
The running time of the procedure is $O(\sum_{a\in \Gamma_1(s)}\left| \dest(s,a,b)\right|)$ (again by simple enumeration).
It follows from the description of $\bb(s,b)$ that as long as the level of the state $s$ is fixed we only decrease 
$\bb(s,b)$. 
Hence we will run $\Remove(s,b)$ at most $n$ times, once for each level of $s$.

\end{enumerate}

\smallskip\noindent{\bf The informal description of the algorithm.}
The informal description of the algorithm is as follows. 
In the initialization phase first all states $s$ are assigned level $\ell(s)=n$,
and then every state is processed using the procedure $\Preprocessing(s)$.
The algorithm is an iterative one and in every iteration executes the following steps 
(unless a fixpoint is reached).
It first considers the set of states $s$ such that $\Good(s)$ is empty and 
decrements the level of $s$. 
If the level of a state reaches~0, then a flag $z$ is assigned to true.
If $z$ is true, then we process every state using the procedure $\Preprocessing$.
Otherwise, for every state $s$ such that $\Good(s)$ is empty, 
the algorithm processes $s$ using $\Preprocessing(s)$; updates $\bb(t,b)$ for all 
predecessors $t$ of $s$ and removes an action when the $\bb(t,b)$ count reaches zero.
The algorithm reaches a fixpoint when the level of no state has changed (the algorithm
keeps track of this with a flag $c$).
The algorithm outputs $\wt{X}^*$ which is the set of states $s$ with strictly positive
level (i.e., $\ell(s)>0$ at the end of the execution).
The formal description of the algorithm is presented in Figure~\ref{fig:solve},
and we refer to the algorithm as \ial.
We first present the runtime analysis and then present the correctness argument.

\smallskip\noindent{\bf Runtime analysis.}
As described above other than the initialization phase, whenever the procedure 
$\Preprocessing(s)$ is run, the level of the state $s$ has decreased or the level of some  other
state has reached~0.
Hence for every state $s$, the procedure can run at most $2\cdot n$ times.
Therefore the total running time for all $\Preprocessing$ operations over all iterations 
is $O(n\cdot |\delta|)$, where $|\delta|=\sum_{s\in S}\sum_{a\in \Gamma_1(s)}\sum_{b\in \Gamma_2(s)}\left|\dest(s,a,b)\right|$.
The procedure $\Remove(s,b)$ is invoked when $\bb(s,b)$ reaches zero, 
and as long as the level of $s$ is fixed the count $\bb(s,b)$ can only decrease.
This implies that we run $\Remove(s,b)$ at most $n$ times, once for each level of $s$.
Hence $O(n\cdot |\delta|)$ is the total running time of operation $\Remove$ over all iterations.\footnote{The running time assumes a data structure that for a given $s$ can find the set 
$\Pred(s)=\{(t,a,b)\mid s\in \dest(t,a,b)\}$ of predecessors in time $O(\left|\Pred(s)\right|)$, which can be easily achieved with a linked list data structure.}

\smallskip\noindent{\bf Correctness analysis.}
We will now present the correctness analysis in the following lemma.

\begin{figure}
\begin{center}

\begin{algorithm}[H]
\lFor{$s\in S$} {$\ell(s)\leftarrow n$\;}
\lFor{$s\in S$} {
$\Preprocessing(s)$\;
}
$c\leftarrow true$\;

\While{$c=true$} {
$c\leftarrow false$; $z\leftarrow false$\;
\For{$s\in S$ st. $\ell(s)>0$ and $\Good(s)=\emptyset$} {
$c\leftarrow true$; $\ell(s)\leftarrow\ell(s)-1$;
\lIf{$\ell(s)=0$}{$z\leftarrow true$\;}
}
\If {$z=true$} {
\lFor{$s\in S$} {
$\Preprocessing(s)$\;
}
}
\Else {
\For{$s\in S$ st. $\ell(s)>0$ and $\Good(s)=\emptyset$} {
$\Preprocessing(s)$\;
\For{$t,a,b$ st. $s\in \dest(t,a,b)$ and $\ell(t)=\ell(s)$} {
$\bb(t,b)\leftarrow \bb(t,b)-1$;
\lIf{$\bb(t,b)=0$} {
$\Remove(t,b)$\;
}}
}
}
}

\Return{$\wt{X}^*=\set{s \mid \ell(s)>0}$}\;
\caption{\ial: Input: Concurrent game structure $G$ with boolean reward function.}
\end{algorithm}

\end{center}
\caption{Improved Algorithm}
\label{fig:solve}
\end{figure}

\begin{lemma}\label{lemm:alg}
Given a concurrent game structure with a boolean reward function as input, 
let $\wt{X}^*$ be the output of algorithm \ial\ (Figure~\ref{fig:solve}).
Then we have $\wt{X}^*=X^*=\Almost_1(\LimInfAvg(1))$.
\end{lemma}
\begin{proof}
We first observe that for the algorithm \ial\ at the end of any iteration, for all $s$ 
the sets $\Allow(s)$, $\Bad(s)$ and $\Good(s)$ are correctly calculated based on the current level 
of all states as defined by the description. 
Let us denote by $\ell^*(s)$ the level of a state $s$ at the end of the execution of the algorithm. 
Recall that $\widetilde{X}^*$ is the set of states $s$ with $\ell^*(s)>0$.
Also recall that $X^*=\nu X. \mu Y. \nu Z. \ASP(X,Y,Z)$.
The correctness proof will show two inclusions. We present them below.

\begin{itemize}

\item \emph{First inclusion: $\wt{X}^* \subseteq X^*$.}
We will show that $\wt{X}^*$ is a fixpoint of the function $f(X)=\mu Y. \nu Z. \ASP(X,Y,Z)$.
Let $\wt{Y}_0=\emptyset$, and $\wt{Y}_{i}=\set{s \mid \ell^*(s) > n-i}$ for $0 < i < n$.
Then for all $0 < i <n$ and for all $s \in (\wt{Y}_i\setminus \wt{Y}_{i-1})$ we have 
$s\in \nu Z. \ASP(\wt{X}^*,\wt{Y}_{i-1},Z)$. 
The fact that $s\in \ASP(\wt{X}^*,\wt{Y}_{i-1},\wt{Y}_i)$ follows since: 
(i)~$\Allow(s)$ as computed by the algorithm is $\Allow(s,\wt{X}^*)$; 
(ii)~$\Bad(s)$ as computed the algorithm is $\Bad(s,\wt{X}^*,\wt{Y}_{i-1})$; and 
(iii)~$\Good(s)$ as computed by the algorithm is $\Good(s,\wt{X}^*,\wt{Y}_{i-1},\wt{Y}_{i})$.
Hence it follows that $\wt{X}^*$ is a fixpoint of $f(X)=\mu Y. \nu Z. \ASP(X,Y,Z)$.
Since $X^*$ is the greatest fixpoint of $f(X)$ we have that $\wt{X}^* \subseteq X^*$.



\item \emph{Second inclusion: $X^* \subseteq \wt{X}^*$.}
Let $i$ and $s$ be such that $s\in (Y_i\setminus Y_{i-1})$, where $Y_0=\emptyset$, and 
for $i >0$ we have $Y_i=\nu Z. \ASP(X^*,Y_{i-1},Z)$. We will show that $i=n+1-\ell^*(s)$. That implies that $\ell^*(s)>0$, because of the following: We have that $s$ can be in $(Y_j\setminus Y_{j-1})$ for at most one value of $j$, because $Y_{k-1}\subseteq Y_k$ for all $k$. Since $(Y_j\setminus Y_{j-1})$ is non-empty for all $j>0$ till the fixpoint is reached we have  $X^* =Y_n$. Together that gives us that $\ell^*(s)>0$. 

We will first show that $i\geq n+1-\ell^*(s)$. Assume towards contradiction that $\ell^*(s)<n+1-i$. Let $k$ be the first iteration of the algorithm in which some state $t\in (Y_j\setminus Y_{j-1})$ goes from level $n+1-j$ to level $n-j$ (this is well-defined because $s$ must do so in some iteration by assumption). We can WLOG assume that $s$ changes from level $n+1-i$ to $n-i$ in iteration $k$. But at the end of iteration $k-1$, we then have that $\Allow(s,X^*)\subseteq \Allow(s)$ and therefore $\Bad(s)\subseteq \Bad(s,X^*,Y_{i-1})$ and therefore $\Good(s,X^*,Y_{i-1},Y_i) \subseteq \Good(s)$, implying that $\Good(s)$ cannot be empty. Hence $s$ does not change level in iteration $k$. That is a contradiction. 

We will next show that $i\leq n+1-\ell^*(s)$. Assume towards contradiction that $\ell^*(s)>n+1-i$. Let $\ell$ be the highest level for which there is a state $t\in (Y_j\setminus Y_{j-1})$ such that $\ell=\ell^*(t)$ and  $\ell^*(t)>n+1-j$ (since $\ell^*(s)>n+1-i$ this is well defined). We can WLOG assume that $\ell^*(s)=\ell$. By the first part of this proof we have that $\widetilde{X}^*\subseteq X^*$, implying that $\Allow(s)\subseteq \Allow(s,X^*)$. By definition of $\ell$, we then get that $\Bad(s,X^*,Y_{\ell-1})\subseteq \Bad(s)$. Let $U$ be the set of states, such that for all $t\in U$ we have that $\ell^*(t)=\ell$. Since for all $t\in (Y_j\setminus Y_{j-1})$ we have that $j\geq n+1-\ell^*(t)$, we get that $(Y_\ell\setminus Y_{\ell-1})\subset U$ (they are not equal since $(Y_\ell \setminus Y_{\ell-1})$ does not contain $s$). We have that $\Good(t)$ is non-empty for all $t\in U$. This implies that $U\subseteq T$, where $T$ is a fixpoint of $\ASP(X^*,Y_{\ell-1},T)$. But $Y_\ell \subset T$ is the largest such fixpoint by definition. That is a contradiction. 
\end{itemize}
The desired result follows.
\end{proof}

\begin{theorem}
The algorithm \ial\ correctly computes the set $\Almost_1(\LimInfAvg(1))$ for a concurrent game 
structure with boolean reward function in quadratic time (in time $O(n\cdot \left| \trans \right|)$,
where $\left|\trans \right|= \sum_{s\in S}\sum_{a\in \Gamma_1(s)}\sum_{b\in \Gamma_2(s)}\left| \dest(s,a,b)\right|$). 
\end{theorem}

\section{Positive Winning}
In this section we will present qualitative determinacy for positive
winning and then establish the strategy complexity results.

\subsection{Qualitative determinacy}\label{subsec4}

In this section we will present a polynomial time algorithm to compute the set 
$\Positive_1(\LimInfAvg(\lambda))$ and $\Positive_1(\LimSupAvg(\lambda))$ for $\lambda=1$
in concurrent games with boolean reward functions, and the qualitative determinacy 
will also be a consequence of the algorithm. 
Again, like in Section \ref{sec:almost}, we will first present the algorithm as a $\mu$-calculus expression. The algorithm is \[
\Positive_1(\LimInfAvg(1)) = \Positive_1(\LimSupAvg(1)) =
\mu Y. \nu Z. \ASP(S,Y,Z),
\]
where $\ASP(X,Y,Z)$ is as defined in Section \ref{sec:almost}.
Let $Y^*= \mu Y. \nu Z. \ASP(S,Y,Z)$ be the fixpoint. 

\begin{lemma}\label{lemm:not-positive}
There is a stationary strategy $\sigma_2$ for player~2 with patience at most $m$ that ensures that for all states $s\in (S\setminus Y^*)$, all strategies $\sigma_1$ for player~1, we have that $\Pr_s^{\sigma_1,\sigma_2^*}(\overline{\LimSupAvg}_{\leq}(1-c))=1$, where $c=(\frac{\delta_{\min}}{m})^{n-1}\cdot \frac{1}{m}$
\end{lemma}
\begin{proof}
In the proof of Lemma \ref{lemm:not-almost} (Remark~\ref{rem:proof}), 
we presented a witness stationary strategy $\sigma_2^1$ that ensured that the set $\ov{X}_1=(S \setminus \mu Y. \nu Z. \ASP(S,Y,Z))$ 
was never left; and for all states $s\in \ov{X}_1$ and all strategies $\sigma_1$ for player~1 we have $\Pr_s^{\sigma_1,\sigma_2^*}(\overline{\LimSupAvg}_{\leq}(1-c))=1$. 
But notice that $(S\setminus Y^*)=\ov{X}_1$. 
Note also that $\sigma_2^1$ played uniformly over some subset of actions in $\Gamma_2(s)$ for any $s\in \ov{X}_1$. 
Hence, the patience of $\sigma_2^1$ is at most $m$.
\end{proof}

\begin{lemma}\label{lemm:positive}
There is a Markov strategy $\sigma_1^*$ for player~1 that ensures that for all states $s\in Y^*$ and all strategies $\sigma_2$ for player~2, we have that $\Pr_s^{\sigma_1^*,\sigma_2}(\LimInfAvg(1))>0$.

\end{lemma}
\begin{proof}
Let $Y_0=\emptyset$ and $Y_{i+1}=\nu Z. \ASP(S,Y_i,Z)$. Also let $\ell$ be the smallest number such that $Y_{\ell+1}=Y_{\ell}$ and $Y^*=Y_{\ell}$. 
To construct $\sigma_1^*$ we will first define a strategy $\sigma_1^{\epsilon}$, for all $\epsilon>0$. 
Fix $\epsilon>0$ and we define $\sigma_1^{\epsilon}$ as follows:
For $s\not \in Y^*$ the strategy plays arbitrarily. For $s\in (Y_i\setminus Y_{i-1})$ the strategy is as follows: \[\sigma_1^{\epsilon}(s)(a)=\begin{cases} \frac{1-\epsilon}{\Gd(s)} & \text{for }a\in \Good(s,S,Y_{i-1},Y_i)\text{ and } \Gd(s)\neq \Aw(s)\\
\frac{1}{\Gd(s)} & \text{for }a\in \Good(s,S,Y_{i-1},Y_i) \text{ and } \Gd(s)= \Aw(s)\\
\frac{\epsilon}{\Aw(s)-\Gd(s)} &\text{for }a\not \in \Good(s,S,Y_{i-1},Y_i),\end{cases}\]
where $\Gd(s)=\left|\Good(s,S,Y_{i-1},Y_i)\right|$ and $\Aw(s)=\left|\Allow(s,S)\right|=\left|\Gamma_1(s)\right|$. By definition of $Y_i$, the set $\Good(s,S,Y_{i-1},Y_i)$ is not empty and hence this is well-defined.

The construction of the desired strategy $\sigma_1^*$ is as follows: 
consider the sequence $\epsilon_1,\epsilon_2,\ldots$ such that 
$\epsilon_1=\frac{1}{4}$ and $\epsilon_{i+1}=\frac{\epsilon_i}{2}$.
In round $k$, the strategy $\sigma_1^*$ will play according to $\sigma_1^{\epsilon_k}$. Note that this is a Markov strategy.

Let $s\in (Y_{i}\setminus Y_{i-1})$. We will now show the statement using induction in $i$. More precisely, assume that we are in $s$ in round $j$, we will show that either some state in $Y_{i-1}$ is reached with positive probability or $\Pr\nolimits_s^{\sigma_1^*,\sigma_2}(\LimInfAvg(1))\geq \frac{1}{2}$. 

For the base case, $i=1$, notice that $Y_0=\emptyset$. 
Hence we need to show that $\Pr\nolimits_s^{\sigma_1^*,\sigma_2}(\LimInfAvg(1))\geq \frac{1}{2}$. By construction of $\sigma_1^*$, the probability for player~1 to ever play a action outside $\Good(s,S,Y_0,Y_{1})$ is $\sum_{k=j}^\infty \epsilon_k\leq \frac{1}{2}$. If no action outside $\Good(s,S,Y_0,Y_{1})$ is ever played we have by definition of $\Good(s,S,Y_0,Y_{1})$ that $Y_1$ is never left and we will in each step get a reward of 1.

For $i>1$ there are two cases. Either player~2 plays an action in $\Bad(s,S,Y_{i-1})$ with positive probability at some point or not. If not, the argument is identical to the base case (except that it is $Y_i$ that will not be left with probability greater than a half). Otherwise, $Y_{i-1}$ is reached with positive probability because all actions are played with positive probability by $\sigma_1^*$ and the statement follows by induction.
\end{proof}

\begin{theorem}[Qualitative determinacy and polynomial time computability]
The following assertions hold for all concurrent game structures with boolean reward functions:
\begin{enumerate}

\item We have 
\[
\begin{array}{rcl}
Y^* & = & \Positive_1(\LimInfAvg(1)) = \Positive_1(\LimSupAvg(1)) \\
& = & 
\displaystyle 
\bigcap_{\varepsilon >0} \Positive_1(\LimInfAvg(1-\varepsilon)) =
\bigcap_{\varepsilon >0} \Positive_1(\LimSupAvg(1-\varepsilon));
\end{array}
\] 
and 
\[
\begin{array}{rcl}
(S\setminus Y^*) & = & 
\Almost_2(\ov{\LimInfAvg}(1)) = \Almost_2(\ov{\LimSupAvg}(1)) \\
& = & 
\displaystyle
\bigcup_{c>0} \Almost_2(\ov{\LimInfAvg}(1-c)) 
=\bigcup_{c>0} \Almost_2(\ov{\LimSupAvg}(1-c));
\end{array}
\]
where $Y^*=\mu Y. \nu Z. \ASP(S,Y,Z)$.

\item The set $Y^*$ can be computed in quadratic time 
(in time $O(n \cdot |\trans|)$) where $|\trans|=\sum_{s\in S}\sum_{a\in \Gamma_1(s)}\sum_{b\in \Gamma_2(s)}\left| \dest(s,a,b)\right|$,
by the straight-forward  computation of the $\mu$-calculus formula
$\mu Y. \nu Z. \ASP(S,Y,Z)$.

\end{enumerate}

\end{theorem}

\begin{proof}
The proof of the theorem is analogous to Theorem \ref{thm:almost}, and uses
Lemma~\ref{lemm:positive} and Lemma~\ref{lemm:not-positive}.
\end{proof}

\subsection{Strategy complexity}\label{subsec5}
In this section we will establish the complexities of the witness 
positive and almost-sure winning strategies for player~1 and player~2,
from their respective winning sets.

Let $0\leq \epsilon<1$ be given. We will show that there exists games with a state $s$ such that there exists $\sigma_1^*$ such that
for all $\sigma_2$ we have $\Pr_s^{\sigma_1^*,\sigma_2}(\LimInfAvg(1-\epsilon))>0$, 
but where no strategy with finite memory for player~1 ensures so.

\smallskip\noindent{\em Game with no finite-memory positive winning strategies.}
Let $\bar{G}$ be the following repeated game with absorbing states. The game $\bar{G}$ has $3$ states, $v^0$, $v^1$ and $v$. For $i\in \{0,1\}$, state $v^i$ is absorbing and has only one action for either player, $a^i$ and $b^i$ respectively and where $\cost(v^i,a^i,b^i)=i$. The state $v$  has two actions for either player. The actions are $a_1$ and $a_2$ for player~1 and $b_1$ and $b_2$ for player~2. Also $\cost(v,a_i,b_j)=0$ except for $i=j=2$, for which  $\cost(v,a_2,b_2)=1$. Furthermore $\delta(v,a_i,b_j)=v^0$ for $i\neq j$, $\delta(v,a_1,b_1)=v^{1}$ and $\delta(v,a_2,b_2)=v$. There is an illustration of $\bar G$ in Figure \ref{fig:barG}. Clearly state $v^1$ and state $v$  are in $Y^*$, because
a Markov strategy which in state $v$ in round $j$ plays action $a_1$ with probability $\frac{1}{2^{j+1}}$ and action $a_2$ with the remaining probability ensures 
$\LimInfAvg(1)$ with positive probability.

\begin{lemma}\label{lemm:positive super constant}
No finite-memory strategy $\sigma_1^*$ for player~1 in $\bar{G}$ ensures that for all $\sigma_2$ and for all $0\leq \epsilon<1$ we have $\Pr_v^{\sigma_1^*,\sigma_2}(\LimInfAvg(1-\epsilon))>0$.  
\end{lemma}
\begin{proof}
The proof will be by contradiction. 
Consider $0\leq \epsilon<1$. Assume towards contradiction that a strategy $\sigma_1$ using finite memory for player~1 exists such that for all $\sigma_2$ we have 
$\Pr_v^{\sigma_1,\sigma_2}(\LimInfAvg(1-\epsilon))>0$. We will show that there exists $\sigma_2$ such that $\Pr_v^{\sigma_1,\sigma_2}(\LimInfAvg(1-\epsilon))=0$
to establish the contradiction.

We will divide the memory states of player~1 into two types. The two types are memory states of type 1, where $\sigma_1$ plays $a_2$ with probability 1 and memory states of type 2, where $\sigma_1$ plays $a_2$ with probability less than 1.
The strategy $\sigma_2$ is then to play $b_1$, if, conditioned on the history so far, $\sigma_1$ is in a memory state of type 1, otherwise play $b_2$. 
Let $p$ be the smallest non-zero probability with which $\sigma_1$ plays $a_1$. We see that in each round, if player~1 follows $\sigma_1$ and player~2 follows $\sigma_2$, there is a probability of at least $p$ to reach $v^0$ and otherwise the plays stays in $v$. Clearly, we must therefore reach $v^0$ after some number of steps with probability 1, which will ensure that all remaining rewards are~0. 
Hence $\Pr_v^{\sigma_1,\sigma_2}(\LimInfAvg(1-\epsilon))=0$. This is a contradiction and the desired result follows.
\end{proof}

\begin{figure}
\centering

\begin{tikzpicture}[node distance=3cm]
\matrix (v1) [label=left:$v^1$:,minimum height=1.5em,minimum width=1.5em,matrix of math nodes,nodes in empty cells, left delimiter={.},right delimiter={.}]
{
\\
};
\draw[black] (v1-1-1.north west) -- (v1-1-1.north east);
\draw[black] (v1-1-1.south west) -- (v1-1-1.south east);
\draw[black] (v1-1-1.north west) -- (v1-1-1.south west);
\draw[black] (v1-1-1.north east) -- (v1-1-1.south east);
\matrix (v) [label=left:$v$:,right of=v1,minimum height=1.5em,minimum width=1.5em,matrix of math nodes,nodes in empty cells, left delimiter={.},right delimiter={.}]
{
&\\
&\\
};
\draw[black] (v-1-1.north west) -- (v-1-2.north east);
\draw[black] (v-1-1.south west) -- (v-1-2.south east);
\draw[black] (v-2-1.south west) -- (v-2-2.south east);
\draw[black] (v-1-1.north west) -- (v-2-1.south west);
\draw[black] (v-1-2.north west) -- (v-2-2.south west);
\draw[black] (v-1-2.north east) -- (v-2-2.south east);
\matrix (v0) [label=left:$v^0$:,right of=v,minimum height=1.5em,minimum width=1.5em,matrix of math nodes,nodes in empty cells, left delimiter={.},right delimiter={.}]
{
\\
};
\draw[black] (v0-1-1.north west) -- (v0-1-1.north east);
\draw[black] (v0-1-1.south west) -- (v0-1-1.south east);
\draw[black] (v0-1-1.north west) -- (v0-1-1.south west);
\draw[black] (v0-1-1.north east) -- (v0-1-1.south east);
\draw[->](v1-1-1.center) .. controls +(90:1) and +(45:1) ..  node[midway,above,right] (x) {1} (v1);
\draw[->](v-1-1.center) to  (v1);
\draw[->](v-1-2.center) to[bend left=75]  (v0);
\draw[->](v-2-1.center) to[bend right=75]  (v0);
\draw[->](v-2-2.center) .. controls +(0:1) and +(-45:1.5) ..  node[midway,below,right] (x) {1} (v);
\draw[->](v0-1-1.center) .. controls +(90:1) and +(45:1) ..   (v0);
\end{tikzpicture}

\caption{The figure shows $\bar G$ that will be used to show infinite-memory requirement for positive winning strategies.}\label{fig:barG}
\end{figure}
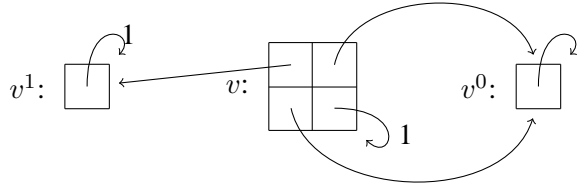

For completeness we will now show that there exists games with states $s$ such that there exists $\sigma_2$ such that for all $\sigma_1$
we have $\Pr_s^{\sigma_1,\sigma_2}(\overline{\LimInfAvg}(1))=1$, but where no stationary strategy $\sigma_2$ with patience less than $m$ exists.

Let $m$ be some fixed number. The game $\underline{G^m}$ has $1$ state, $v$. The state $v$  has $m$ actions for both players. The actions are $a_1,a_2,\dots,a_m$ for player~1 and $b_1,b_2,\dots,b_m$ for player~2. Also $\cost(v,a_i,b_j)=1$ for $i\neq j$ and $\cost(v,a_i,b_i)=0$. Furthermore $\delta(v,a_i,b_j)=v$. There is an illustration of $\underline {G^2}$ in Figure \ref{fig:underG}. Clearly state $v$ is in $(S\setminus Y^*)$.

\begin{lemma}\label{lemm:positive patience lower bound}
For all $m>0$, no stationary strategy $\sigma_2$ for player~2 with patience less than $m$ in $\underline{G^m}$ ensures that for all $\sigma_1$ we have $\Pr_v^{\sigma_1,\sigma_2}(\overline{\LimInfAvg}(1))=1$.
\end{lemma}

\begin{proof}
The proof will be by contradiction. Assume that such a strategy $\sigma_2$ for player~2 exists. Clearly it must play some action $b_i$ with probability 0. Hence, if $\sigma_1$ plays $a_i$ with probability~1, 
we have $\Pr_v^{\sigma_1,\sigma_2}(\LimInfAvg(1))=1$. That is a contradiction.
\end{proof}

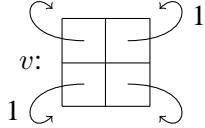
\begin{figure}
\centering
\begin{tikzpicture}[node distance=3cm]
\matrix (v) [label=left:$v$:,minimum height=1.5em,minimum width=1.5em,matrix of math nodes,nodes in empty cells, left delimiter={.},right delimiter={.}]
{
&\\
&\\
};
\draw[black] (v-1-1.north west) -- (v-1-2.north east);
\draw[black] (v-1-1.south west) -- (v-1-2.south east);
\draw[black] (v-2-1.south west) -- (v-2-2.south east);
\draw[black] (v-1-1.north west) -- (v-2-1.south west);
\draw[black] (v-1-2.north west) -- (v-2-2.south west);
\draw[black] (v-1-2.north east) -- (v-2-2.south east);
\draw[->](v-1-1.center) .. controls +(180:1) and +(135:1.5) ..   (v);
\draw[->](v-1-2.center) .. controls +(0:1) and +(45:1.5) ..  node[midway,above,right] (x) {1} (v);
\draw[->](v-2-1.center) .. controls +(180:1) and +(-135:1.5) ..  node[midway,below,left] (x) {1} (v);
\draw[->](v-2-2.center) .. controls +(0:1) and +(-45:1.5) ..   (v);
\end{tikzpicture}

\caption{The example shows $\underline G^2$ of the game family $\underline G^n$
that will be used to show that patience $m$ is required by player~2.\label{fig:underG}}
\end{figure}

\begin{theorem}[Strategy complexity]
The following assertions hold for concurrent games with boolean reward functions:
\begin{enumerate}

\item Positive winning strategies for player~1 for objectives 
$\LimInfAvg(1)$ and $\LimInfAvg(1-\epsilon)$, for $\epsilon>0$, 
(also $\LimSupAvg(1)$ and $\LimSupAvg(1-\epsilon)$, for $\epsilon>0$)
require infinite-memory in general; whenever such positive winning strategies exist, 
Markov strategies are sufficient and the optimal bound for time-dependent memory is $T$, 
for all rounds $T>0$.

\item Stationary almost-sure winning strategies for player~2 exist for objectives 
$\ov{\LimInfAvg}(1)$ and $\ov{\LimInfAvg}(1-c)$, for some constant $c>0$, 
(also $\ov{\LimSupAvg}(1)$ and $\ov{\LimSupAvg}(1-c)$, for some constant $c>0$)
 and the optimal bound for patience of such stationary almost-sure winning 
strategies is the size of the action space.

\end{enumerate}
\end{theorem}

\begin{proof}
The proofs are as follows:
\begin{enumerate}

\item Lemma~\ref{lemm:positive super constant} shows that infinite-memory is required, and 
Lemma~\ref{lemm:positive}  shows that Markov strategies are sufficient for positive
winning.
The sufficiency of Markov strategies establishes the $T$ upper bound 
for time-dependent memory; and Lemma~\ref{lemm:positive super constant} and Lemma~\ref{lemm:log T} establish
the $T$ lower bound for time-dependent memory.

\item The existence of stationary positive winning strategies with $m$ patience  follows from Lemma~\ref{lemm:not-positive}.
The lower bound of $m$ for patience follows from Lemma~\ref{lemm:positive patience lower bound}.
\end{enumerate}
The desired result follows.
\end{proof}

\section{Almost and Positive Winning for Quantitative Path Constraints}\label{Sec:quan}

In this section our goal is to establish hardness results for polynomial-time 
computability of $\Almost_1(\LimInfAvg(\lambda))$ and $\Positive_1(\LimInfAvg(\lambda))$,
given $\lambda$ is a rational number in the interval $(0,1)$, 
for turn-based stochastic games with boolean reward functions.
We first mention several related polynomial-time computability results: 
(1)~turn-based deterministic games with boolean reward functions can be 
solved in polynomial time (follows from~\cite{ZP96} as the pseudo-polynomial
time algorithm is polynomial for boolean rewards);
(2)~turn-based stochastic reachability games can be solved in polynomial
time for almost-sure and positive winning (follows from the results of~\cite{CJH03}
that show a polynomial reduction to turn-based deterministic B\"uchi games for almost-sure
and positive winning);
and
(3)~turn-based stochastic and concurrent stochastic games can be solved
in polynomial time if $\lambda=1$ as established in the 
previous sections for almost-sure and positive winning.
Hence our hardness result for almost-sure and positive winning for 
turn-based stochastic boolean reward games with $\lambda\neq 1$ is tight in
the sense that relaxation to deterministic games, or reachability objectives,
or $\lambda=1$ ensures polynomial-time computability.
Our hardness result will be a reduction from the problem of deciding if $\val(s)\geq c$, 
given a constant $c\geq 0$ and a state $s$ in a turn-based deterministic mean-payoff game 
with \emph{arbitrary} rewards to the problem of deciding whether $t\in \Almost_1(\LimInfAvg(\lambda))$ 
in turn-based stochastic games with \emph{boolean} rewards, for $\lambda \in (0,1)$.
Our reduction will also ensure that in the game obtained we have 
$\Almost_1(\LimInfAvg(\lambda))=\Positive_1(\LimInfAvg(\lambda))$. 
Hence the hardness also follows for the problem of deciding whether $t\in \Positive_1(\LimInfAvg(\lambda))$ 
in turn-based stochastic games with boolean rewards.
The polynomial-time computability of optimal values in turn-based deterministic mean-payoff games 
with arbitrary rewards is a long-standing open problem (the decision problem is in 
NP $\cap$ coNP, but no deterministic sub-exponential time algorithm is known).
To present the reduction we first present an equivalent and convenient notation 
for turn-based deterministic and turn-based stochastic games.

\smallskip\noindent{\bf Equivalent convenient notation for turn-based games.}
An equivalent formulation for turn-based stochastic games is as follows:
a turn-based stochastic game $G=((S,E),(S_1,S_2,S_P),\trans)$ 
consists of a finite set $S$ of states, $E$ of edges, a partition of the 
state space into player~1, player~2 and probabilistic states,
($S_1,S_2,S_P$, respectively) and a probabilistic transition function $\trans:S_P \to \distr(S)$ 
such that for all $s \in S_P$ and $t \in S$ we have $(s,t) \in E$ iff
$\trans(s)(t)>0$. 
In a turn-based stochastic game, in player~1 states the successor state is chosen
by player~1 and likewise for player~2 states. In probabilistic states
the successor state is chosen according to the probabilistic transition
function $\trans$. 
For a turn-based deterministic game we have $S_P=\emptyset$, and hence 
we do not need the transition function $\trans$, and simply represent them 
as $G=((S,E),(S_1,S_2))$.

\smallskip\noindent{\bf Optimal values in DMPGs.}
A DMPG (deterministic mean-payoff game) consists of a turn-based deterministic game 
$G=((S,E),(S_1,S_2))$ with a reward function $\cost:E \to \set{0,1,\ldots,M}$,
(note that the rewards are non-negative integers and not necessarily boolean).
The \emph{optimal value} for a state $s$, denoted as $\val(s)$, is the maximal limit-inf-average 
value that player~1 can ensure with a positional strategy against all positional 
strategies of the opponent.
Formally, given two positional strategies $\sigma_1$ and $\sigma_2$, and a starting 
state $s$, an unique cycle $C$ is executed infinitely often, and the mean-payoff 
value for $\sigma_1$ and $\sigma_2$ from $s$, denoted $\val(s,\sigma_1,\sigma_2)$,
is $\frac{\sum_{e\in C} \cost(e)}{|C|}$, the sum of the rewards in $C$, divided by
the length of $C$.
Then $\val(s)$ is the value of state $s\in S$, that is, $\val(s)=\max_{\sigma_1} \min_{\sigma_2} \val(s,\sigma_1,\sigma_2)$,
where $\sigma_1$ and $\sigma_2$ ranges over positional strategies 
of player~1 and player~2, respectively.
Given a rational number $\lambda$, the decision problem of whether 
$\val(s) \geq \lambda$ lies in NP $\cap$ coNP and can be computed in pseudo-polynomial 
time (in time $O(|S| \cdot |E| \cdot M)$) for DMPGs~\cite{ZP96,Brim11}. Finding an algorithm that runs in polynomial time and solves that decision problem is a long-standing open problem. 
We will reduce the computation of the value problem for DMPGs to almost-sure winning 
in turn-based stochastic games with boolean rewards but quantitative path constraints, 
i.e., $\Almost_1(\LimInfAvg(\lambda))$, for $\lambda \in (0,1)$.

\newcommand{\red}{\mathsf{Red}}

\smallskip\noindent{\bf Reduction.} 
Given the DMPG $G=((S,E),(S_1,S_2))$ with non-negative integral rewards, with largest reward $M$, the construction 
of a turn-based stochastic game $G'=((S',E'),(S_1',S_2',S_P),\trans')$ is as follows:
$S_1'=S_1$ and $S_2'=S_2$; and for every edge $e=(s,t)$ in $G$ with reward $\cost(e)$, we replace the edge between $s$ and $t$ 
 with a gadget with four additional states (namely, $v^1(e),v^2(e),v^3(e)$ and $v^4(e)$) 
along with $s$ and $t$ and eight edges with boolean rewards. 
The gadget is as follows: each state $v^i(e)$ is a probabilistic state for $i\in \{1,2,3,4\}$, and
from state $s$ there is an edge corresponding to the edge $e$ that goes to $v^1(e)$.
The other transitions from the states in the gadget are specified below:
(i)~from $v^1(e)$ the next state is state $v^2(e)$ with probability $\frac{\cost(e)}{M}$ and state $v^3(e)$ with probability $\frac{M-\cost(e)}{M}$, and 
the edges have reward~0;
(ii)~the next state from either state $v^2(e)$ or state $v^3(e)$ is $v^4(e)$ with probability $1-\frac{1}{M}$ and state $t$ with probability $\frac{1}{M}$, and the edges from $v^2(e)$ have reward~1 and
the edges from $v^3(e)$ have reward~0; and 
(iii)~the next state from state $v^4(e)$ is $v^1(e)$ with edge reward~0. 
There is a illustration of the gadget in Figure \ref{fig:gadget}.
We refer to the boolean reward turn-based stochastic game obtained by the above
reduction from a DMPG $G$ as $G'=\red(G)$.
Also note that the reduction is polynomial as all the probabilities can 
be expressed in polynomial size given the input DMPG $G$.

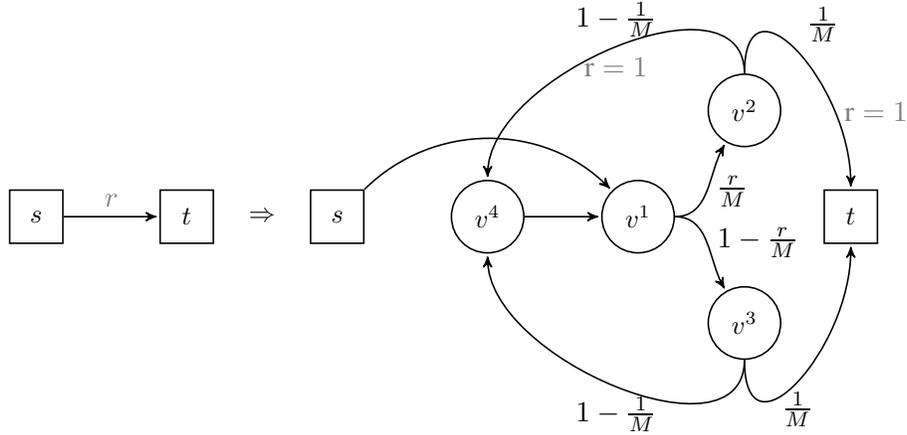
\begin{figure}
\begin{centering}

\begin{tikzpicture}[->,>=stealth',shorten >=1pt,auto,node distance=2cm,
                    semithick ]
      
\tikzstyle{every state}=[fill=white,draw=black,text=black,font=\small , inner sep=-.05 cm]
\tikzstyle{square}=[state,regular polygon,regular polygon sides=4,minimum height=1cm,minimum width=1cm]

    \node[square] (s){$s$};
    \node[square,right of=s] (t){$t$};
    \node[state,draw=white,right of=t,node distance=1cm] (ra){$\Rightarrow$};
    \node[square,right of=ra,node distance=1cm] (s2){$s$};
    \node[state,right of=s2] (v4){$v^4$};
    \node[state,right of=v4] (v1){$v^1$};
    \node[state,above right  of=v1] (v2){$v^2$};
    \node[state,below right  of=v1] (v3){$v^3$};
    \node[square,below right  of=v2] (t2){$t$};
    \path (s) edge node[gray] {$r$} (t)
    (s2) edge[bend left,out=45,in=135] (v1)
    (v4) edge (v1)
    (v1) edge[out=0,in=-120] node [midway, right] {$\frac{r}{M}$} (v2)
    (v1) edge[out=0,in=120] node [midway, right] {$1-\frac{r}{M}$} (v3)
    (v2) edge[out=90,in=90,distance=40] node[below,midway,gray] {$\cost=1$} node[midway,above] {$1-\frac{1}{M}$} (v4)
    (v2) edge[out=90,in=90,distance=40] node[right,near end,gray] {$\cost=1$} node[midway] {$\frac{1}{M}$} (t2)
    (v3) edge[out=-90,in=-90,distance=40] node[midway,below ] {$1-\frac{1}{M}$} (v4)
    (v3) edge[out=-90,in=-90,distance=40] node[midway,below] {$\frac{1}{M}$} (t2)    
    ;

\end{tikzpicture}

\end{centering}
\caption{Our gadget for reducing a DMPG to a turn-based stochastic boolean reward game: The edges that go to more than one state have the probability annotated on them in black; and 
any non-zero reward is annotated on the corresponding edge in gray.}
\label{fig:gadget}
\end{figure}

\smallskip\noindent{\em Property of the reduction.}
Let $e=(s,t)$ be some edge in $G$ with reward $\cost(e)$. It is easy to verify that in the gadget, if the edge to $v^1(e)$ is taken from $s$ in $G'$, 
then we eventually reach $t$ with probability~1, while we expect to get $\cost(e)$ rewards of value~1 
and $3\cdot M -\cost(e)$ rewards of value~0, before reaching $t$. The expected number of steps to reach $t$ from $s$ is thus always $3\cdot M$ and is independent
of the reward value $\cost(e)$.
Hence the total expected reward is $\cost(e)$ and one step of the game $G$ is simulated by 
$3\cdot M$ steps in $G'$.
We will show that if a state $s$ in $G$ has optimal value $\val(s)$, then the corresponding state in $G'$ is in 
$\Almost_1(\LimInfAvg(\frac{\val(s)}{3M}))$. 
Also, we will show that if a state in $G'$ is in $\Almost_1(\LimInfAvg(\lambda))$, then the corresponding state in $G$ has optimal value 
of at least $3\cdot M\cdot \lambda$. 
We present the results in the following two lemmas.

\smallskip\noindent{\em One basic property of Markov chains.}
In both the lemmas we will use the following \emph{basic property} of a Markov chain.
Consider a Markov chain with arbitrary rewards, and closed recurrent set $C$ of
the Markov chain. 
Let $\alpha$ be the expected mean-payoff value from a starting state $s$ in $C$ 
(the expected mean-payoff value is independent of the start state since $C$ is a 
closed recurrent set).
Then for all $s \in C$, we have $s \in \Almost_1(\LimInfAvg(\alpha))$ and for all
$\alpha'>\alpha$ we have $s \not\in \Positive_1(\LimSupAvg(\alpha'))$, i.e.,
almost-surely the mean-payoff value is at least $\alpha$, and for every $\alpha'>\alpha$ 
the mean-payoff is at least $\alpha'$ with probability~0.
The above basic property result follows by the almost-sure convergence to the invariant 
distribution (or Cesaro limit) for a closed recurrent set of a Markov chain.

\begin{lemma}
Given a DMPG $G$ with largest reward $M$ for a state $s$ in $G$ we have that 
the corresponding state in $G'$ belongs to $\Almost_1(\LimInfAvg(\frac{\val(s)}{3M}))$ 
and $\Positive_1(\LimInfAvg(\frac{\val(s)}{3M}))$, where $G'=\red(G)$.
\end{lemma}
\begin{proof}
Consider an optimal positional (pure and stationary) strategy $\sigma_1$ for player 1 in $G$ 
(such an optimal strategy exists in DMPGs~\cite{EM79}). 
The strategy ensures that $\LimInfAvg$ is at least $\val(s)$ if the play starts in $s$ against any 
strategy for player~2. 
Consider the corresponding strategy $\sigma'_1$ of $\sigma_1$ in $G'$. 
Consider a positional best response strategy $\sigma'_2$ for player~2 in $G'$ to $\sigma'_1$, if the play starts 
in the state that corresponds to state $s$. 
The play in $G'$ given $\sigma_1'$ and $\sigma_2'$ reaches an unique closed recurrent set $C'$ with probability~1 
(i.e., the set $C'$ corresponds to the unique cycle $C$ reachable from $s$ given strategies $\sigma_1$ and
the corresponding strategy $\sigma_2$ of $\sigma_2'$, and the states introduced by the gadget).
We have the following desired properties.
First, in the closed recurrent set $C'$ of $G'$ the expected limit-average payoff is at least $\frac{\val(s)}{3\cdot M}$,
since the average reward of the cycle $C$ in $G$ is at least $\val(s)$, and in $G'$ every step of $G$ is 
simulated by $3\cdot M$ steps with the same total reward value in expectation.
Second, since we have a Markov chain, the expectation and almost-sure satisfaction coincide for 
closed recurrent set (the basic property of Markov chains).
Finally, in $G'$ the closed recurrent set $C'$ is reached with probability~1 given 
the strategies $\sigma_1'$ and $\sigma_2'$, from the starting state corresponding to $s$.
This shows that the corresponding state to $s$ in $G'$ belongs to 
$\Almost_1(\LimInfAvg(\frac{\val(s)}{3M}))$ and $\Positive_1(\LimInfAvg(\frac{\val(s)}{3M}))$.
\end{proof}

\begin{lemma}
Given a DMPG $G$ with largest reward $M$ for a state $s$ in $G$, if the corresponding state 
in $G'$ belongs to either $\Almost_1(\LimInfAvg(\lambda))$ or $\Positive_1(\LimInfAvg(\lambda))$, 
then $\val(s) \geq 3\cdot M \cdot \lambda$, where $G'=\red(G)$.
\end{lemma}
\begin{proof}
Consider a positional strategy for player~1 in $G'$ (such a strategy exists
since we consider turn-based stochastic games) to ensure $\LimInfAvg(\lambda)$ 
with probability~1 from the state corresponding to $s$.
Consider the corresponding strategy $\sigma_1$ in $G$ and a positional best response strategy $\sigma_2$ 
of player~2 in $G$, and consider the corresponding strategy $\sigma_2'$ of $\sigma_2$ in $G'$.
Let the unique cycle executed in $G$ given $\sigma_1$ and $\sigma_2$ from $s$ be $C$.
The unique closed recurrent set reached with probability~1 in $G'$ from $s$ given 
$\sigma_1'$ and $\sigma_2'$ is $C'$. Hence, the set $C'$ consists of the states in $C$ along 
with the gadget states of $C$ the reduction.
Since the state $s$ belongs to  $\Almost_1(\LimInfAvg(\lambda))$ or $\Positive_1(\LimInfAvg(\lambda))$ in $G'$,
it follows from the basic property of Markov chains that the expected average reward of the closed recurrent 
set $C'$ is at least $\lambda$.
Since every step of $G$ is simulated by $3\cdot M$ steps in $G'$ it follows that   
the average reward of the cycle $C$ must be at least $3\cdot M \cdot \lambda$.
This completes the proof.
\end{proof}

The following theorem follows from the two previous lemmas and establishes  the 
desired hardness result.

\begin{theorem}[Hardness for quantitative constraints]
Given a DMPG $G$, a state $s$ and a rational value $\lambda$, 
we have $\val(s) \geq \lambda$ in $G$ iff 
$s \in \Almost_1(\LimInfAvg(\frac{\lambda}{3\cdot M})) =\Positive_1(\LimInfAvg(\frac{\lambda}{3\cdot M}))$
in $G'=\red(G)$. 
\end{theorem}

\section{Discussion and Conclusion}
We first discuss two aspects of our results and then conclude.
We first remark how general rational-valued reward functions 
can be reduced to boolean reward functions for qualitative analysis.
We then remark about the optimality of our algorithm for positive 
winning.

\begin{remark}\label{rem:rew}
For all the results for almost-sure and positive winning with exact and 
limit qualitative constraints, we considered boolean reward functions.
For general rational-valued reward functions without loss of generality we
can consider that the rewards are in the interval $[0,1]$ by shifting 
and scaling of the rewards.
Formally, given a reward function $\cost$, and rational values $x$ and $y$, 
consider the modified reward function $\wh{\cost}=x\cdot (\cost + y)$ that 
assigns to every transition $e$ the reward value $x\cdot(\cost(e) +y)$.
The limit-average value of every path under the modified reward function 
is obtained by first adding $y$ to the limit-average value for the original 
reward function and then multiplying the result by $x$.
Hence given a rational-valued reward function $\cost$, let $[y_\ell,y_u]$ be 
the domain of the function.
The modified reward function $x\cdot(\cost+y)$ with $y=-y_{\ell}$ and 
$x=\frac{1}{(y_u-y_{\ell})}$ is a reward function with domain $[0,1]$.
Observe that all our results for boolean reward functions with 
$\LimInfAvg(1)$ and $\LimInfAvg(1-\epsilon)$, for all $\epsilon>0$ 
(and also for $\LimSupAvg$) also hold for reward functions with rewards 
in the interval $[0,1]$, since in our proof we can replace reward~0 by the 
maximal reward that is strictly less than~1.
Hence our results also extend to rational-valued reward functions where 
the objective is to ensure the maximal reward value.
\end{remark}

\begin{remark}
Note that both for positive and almost-sure winning with exact and limit 
qualitative constraints we have presented quadratic time algorithms.
For almost-sure winning the bound of our algorithm matches the best known
bound for the special case of concurrent reachability games.
For positive winning, concurrent reachability games can be solved in 
linear time.
However for the special case of turn-based deterministic mean-payoff games 
with boolean rewards, the almost-sure and the positive winning sets for the exact 
and the limit qualitative constraints coincide with the winning set for coB\"uchi games, 
where the goal is to ensure that eventually a set $C$ is reached and never left 
(the set $C$ is the set of states with reward~1). 
The current best known algorithms for turn-based deterministic coB\"uchi games
are quadratic~\cite{CH12}, and our algorithm matches the quadratic bound known for the special
case of turn-based deterministic games.
\end{remark}

\smallskip\noindent{\bf Concluding remarks.}
In this work we considered qualitative analysis of concurrent mean-payoff games.
For qualitative constraints, we established the qualitative determinacy results;
presented quadratic algorithms to compute almost-sure and positive winning sets
(matching the best known bounds for the simpler case of reachability objectives 
or turn-based deterministic games); and presented a complete characterization of 
the strategy complexity.
We established a hardness result for qualitative analysis with quantitative path constraints.

\bibliographystyle{plain}
\bibliography{diss}

\begin{thebibliography}{10}

\bibitem{AHK02}
R.~Alur, T.A. Henzinger, and O.~Kupferman.
\newblock Alternating-time temporal logic.
\newblock {\em Journal of the ACM}, 49:672--713, 2002.

\bibitem{BGB12}
C.~Baier, M.~Gr{\"o}{\ss}er, and N.~Bertrand.
\newblock Probabilistic omega-automata.
\newblock {\em Journal of the ACM}, 59(1), 2012.

\bibitem{BGG09}
N.~Bertrand, B.~Genest, and H.~Gimbert.
\newblock Qualitative determinacy and decidability of stochastic games with
  signals.
\newblock In {\em Proc. of LICS}, pages 319--328, 2009.

\bibitem{BK76}
T.~Bewley and E.~Kohlberg.
\newblock The asymptotic behavior of stochastic games.
\newblock {\em Math Oper Research}, (1), 1976.

\bibitem{BF68}
D.~Blackwell and T.S. Ferguson.
\newblock The big match.
\newblock {\em Annals of Mathematical Statistics}, 39:159--163, 1968.

\bibitem{BCHJ09}
R.~Bloem, K.~Chatterjee, T.~A. Henzinger, and B.~Jobstmann.
\newblock Better quality in synthesis through quantitative objectives.
\newblock In {\em CAV}, pages 140--156, 2009.

\bibitem{BBFR13}
A.~Bohy, V.~Bruy{\`e}re, E.~Filiot, and J-F. Raskin.
\newblock Synthesis from ltl specifications with mean-payoff objectives.
\newblock In {\em TACAS}, pages 169--184, 2013.

\bibitem{BCHK11}
U.~Boker, K.~Chatterjee, T.~A. Henzinger, and O.~Kupferman.
\newblock Temporal specifications with accumulative values.
\newblock In {\em LICS}, pages 43--52, 2011.

\bibitem{BBKO11}
T.~Br{\'a}zdil, V.~Brozek, A.~Kucera, and J.~Obdrz{\'a}lek.
\newblock Qualitative reachability in stochastic bpa games.
\newblock {\em Inf. Comput.}, 209(8):1160--1183, 2011.

\bibitem{Brim11}
L.~Brim, J.~Chaloupka, L.~Doyen, R.~Gentilini, and J-F. Raskin.
\newblock {Faster algorithms for mean-payoff games}.
\newblock {\em Formal Methods in System Design}, 38(2):97--118, 2011.

\bibitem{CCHRS11}
P.~Cern{\'y}, K.~Chatterjee, T.~A. Henzinger, A.~Radhakrishna, and R.~Singh.
\newblock Quantitative synthesis for concurrent programs.
\newblock In {\em CAV}, pages 243--259, 2011.

\bibitem{CCT13}
K.~Chatterjee, M.~Chmelik, and M.~Tracol.
\newblock What is decidable about partially observable {M}arkov decision
  processes with omega-regular objectives.
\newblock In {\em Proceedings of CSL 2013: Computer Science Logic}, 2013.

\bibitem{CD12}
K.~Chatterjee and L.~Doyen.
\newblock Partial-observation stochastic games: How to win when belief fails.
\newblock In {\em LICS}, 2012.

\bibitem{CDH10}
K.~Chatterjee, L.~Doyen, and T.~A. Henzinger.
\newblock Quantitative languages.
\newblock {\em ACM Trans. Comput. Log.}, 11(4), 2010.

\bibitem{CDHR07}
K.~Chatterjee, L.~Doyen, T.~A. Henzinger, and J.-F. Raskin.
\newblock Algorithms for omega-regular games of incomplete information.
\newblock {\em Logical Methods in Computer Science}, 3(3:4), 2007.

\bibitem{CDNV14}
K.~Chatterjee, L.~Doyen, S.~Nain, and M.~Y. Vardi.
\newblock The complexity of partial-observation stochastic parity games with
  finite-memory strategies.
\newblock In {\em FoSSaCS}, pages 242--257, 2014.

\bibitem{CH11}
K.~Chatterjee and M.~Henzinger.
\newblock Faster and dynamic algorithms for maximal end-component decomposition
  and related graph problems in probabilistic verification.
\newblock In {\em SODA}, pages 1318--1336, 2011.

\bibitem{CH12}
K.~Chatterjee and M.~Henzinger.
\newblock An {{\it O}({\it n}$^{\mbox{2}}$)} time algorithm for alternating
  {B{\"u}chi} games.
\newblock In {\em SODA}, pages 1386--1399, 2012.

\bibitem{CHJS10}
K.~Chatterjee, T.~A. Henzinger, B.~Jobstmann, and R.~Singh.
\newblock Measuring and synthesizing systems in probabilistic environments.
\newblock In {\em CAV}, pages 380--395, 2010.

\bibitem{CJH03}
K.~Chatterjee, M.~Jurdzi{\'n}ski, and T.A. Henzinger.
\newblock Simple stochastic parity games.
\newblock In {\em CSL'03}, volume 2803 of {\em LNCS}, pages 100--113. Springer,
  2003.

\bibitem{CMH08}
K.~Chatterjee, R.~Majumdar, and T.~A. Henzinger.
\newblock Stochastic limit-average games are in {EXPTIME}.
\newblock {\em International Journal Game Theory}, 37(2):219--234, 2008.

\bibitem{CT12}
K.~Chatterjee and M.~Tracol.
\newblock Decidable problems for probabilistic automata on infinite words.
\newblock In {\em LICS}, pages 185--194, 2012.

\bibitem{Con92}
A.~Condon.
\newblock The complexity of stochastic games.
\newblock {\em Information and Computation}, 96(2):203--224, 1992.

\bibitem{EMSOFT05}
L.~de~Alfaro, M.~Faella, R.~Majumdar, and V.~Raman.
\newblock Code-aware resource management.
\newblock In {\em EMSOFT 05}. ACM, 2005.

\bibitem{dAHK98}
L.~de~Alfaro, T.A. Henzinger, and O.~Kupferman.
\newblock Concurrent reachability games.
\newblock In {\em FOCS'98}, pages 564--575. IEEE, 1998.

\bibitem{AHM00a}
L.~de~Alfaro, T.A. Henzinger, and F.Y.C. Mang.
\newblock The control of synchronous systems.
\newblock In {\em CONCUR'00}, LNCS 1877, pages 458--473. Springer, 2000.

\bibitem{AHM01a}
L.~de~Alfaro, T.A. Henzinger, and F.Y.C. Mang.
\newblock The control of synchronous systems, part ii.
\newblock In {\em CONCUR'01}, LNCS 2154, pages 566--580. Springer, 2001.

\bibitem{DM12}
M.~Droste and I.~Meinecke.
\newblock Weighted automata and weighted mso logics for average and long-time
  behaviors.
\newblock {\em Inf. Comput.}, 220:44--59, 2012.

\bibitem{EM79}
A.~Ehrenfeucht and J.~Mycielski.
\newblock Positional strategies for mean payoff games.
\newblock {\em International Journal of Game Theory}, 8(2):109--113, 1979.

\bibitem{EY05}
K.~Etessami and M.~Yannakakis.
\newblock Recursive {M}arkov decision processes and recursive stochastic games.
\newblock In {\em ICALP'05}, LNCS 3580, Springer, pages 891--903, 2005.

\bibitem{EY06}
K.~Etessami and M.~Yannakakis.
\newblock Recursive concurrent stochastic games.
\newblock In {\em ICALP'06 (2)}, LNCS 4052, Springer, pages 324--335, 2006.

\bibitem{Eve57}
H.~Everett.
\newblock Recursive games.
\newblock In {\em Contributions to the Theory of Games {III}}, volume~39 of
  {\em Annals of Mathematical Studies}, pages 47--78, 1957.

\bibitem{FV97}
J.~Filar and K.~Vrieze.
\newblock {\em Competitive {Markov} Decision Processes}.
\newblock Springer-Verlag, 1997.

\bibitem{Gil57}
D.~Gillette.
\newblock Stochastic games with zero stop probabilitites.
\newblock In {\em Contributions to the Theory of Games III}, pages 179--188.
  Princeton University Press, 1957.

\bibitem{GKK88}
V.~A. Gurvich, A.~V. Karzanov, and L.~G. Khachiyan.
\newblock {Cyclic games and an algorithm to find minimax cycle means in
  directed graphs}.
\newblock {\em USSR Comput. Math. Math. Phys.}, 28(5):85--91, April 1990.

\bibitem{HIM11}
K.~A. Hansen, R.~Ibsen-Jensen, and P.~B. Miltersen.
\newblock The complexity of solving reachability games using value and strategy
  iteration.
\newblock In {\em CSR}, pages 77--90, 2011.

\bibitem{HKLMT11}
K.~A. Hansen, M.~Kouck{\'y}, N.~Lauritzen, P.~B. Miltersen, and E.~P.
  Tsigaridas.
\newblock Exact algorithms for solving stochastic games: extended abstract.
\newblock In {\em STOC}, pages 205--214, 2011.

\bibitem{HKM09}
K.~A. Hansen, M.~Kouck{\'y}, and P.~B. Miltersen.
\newblock Winning concurrent reachability games requires doubly-exponential
  patience.
\newblock In {\em LICS}, pages 332--341, 2009.

\bibitem{HK66}
A.J. Hoffman and R.M. Karp.
\newblock On nonterminating stochastic games.
\newblock {\em Management Sciences}, 12(5):359--370, 1966.

\bibitem{Jur00}
M.~Jurdzinski.
\newblock Small progress measures for solving parity games.
\newblock In {\em STACS'00}, pages 290--301. LNCS 1770, Springer, 2000.

\bibitem{kohlberg74}
Elon Kohlberg.
\newblock Repeated games with absorbing states.
\newblock {\em The Annals of Statistics}, 2(4):pp. 724--738, 1974.

\bibitem{KNP_PRISM00}
M.~Kwiatkowska, G.~Norman, and D.~Parker.
\newblock Verifying randomized distributed algorithms with prism.
\newblock In {\em Workshop on Advances in Verification (WAVE'00)}, 2000.

\bibitem{LigLip69}
T.~A. Liggett and S.~A. Lippman.
\newblock Stochastic games with perfect information and time average payoff.
\newblock {\em Siam Review}, 11:604--607, 1969.

\bibitem{MN81}
J.F. Mertens and A.~Neyman.
\newblock Stochastic games.
\newblock {\em International Journal of Game Theory}, 10:53--66, 1981.

\bibitem{NV13}
S.~Nain and M.~Y. Vardi.
\newblock Solving partial-information stochastic parity games.
\newblock In {\em LICS}, pages 341--348, 2013.

\bibitem{PSL00}
A.~Pogosyants, R.~Segala, and N.~Lynch.
\newblock Verification of the randomized consensus algorithm of {A}spnes and
  {H}erlihy: a case study.
\newblock {\em Distributed Computing}, 13(3):155--186, 2000.

\bibitem{Puterman}
M.L. Puterman.
\newblock {\em {Markov} Decision Processes}.
\newblock John Wiley and Sons, 1994.

\bibitem{Sha53}
L.S. Shapley.
\newblock Stochastic games.
\newblock {\em Proc.\ Nat.\ Acad.\ Sci. {USA}}, 39:1095--1100, 1953.

\bibitem{Sto02b}
M.I.A. Stoelinga.
\newblock Fun with {FireWire}: Experiments with verifying the {IEEE1394} root
  contention protocol.
\newblock In {\em Formal Aspects of Computing}, 2002.

\bibitem{VardiP85}
M.Y. Vardi.
\newblock Automatic verification of probabilistic concurrent finite-state
  systems.
\newblock In {\em FOCS'85}, pages 327--338. IEEE Computer Society Press, 1985.

\bibitem{ZP96}
U.~Zwick and M.~Paterson.
\newblock The complexity of mean payoff games on graphs.
\newblock {\em Theoretical Computer Science}, 158:343--359, 1996.

\end{thebibliography}

\end{document}